\RequirePackage{fixltx2e}
\documentclass[reprint, amssymb, amsmath, amsfonts, aip, cha, floatfix, showpacs, twoside]{revtex4-1}

\usepackage{amsthm, calc, xspace, enumerate, graphicx}

\newtheorem{theorem}{\textbf{Theorem}}
\newtheorem{definition}{\textbf{Definition}}

\newcommand{\abs}[1]{\mathopen{}\mathclose\bgroup\left\lvert#1\aftergroup\egroup\right\rvert}
\newcommand{\norm}[1]{\mathopen{}\mathclose\bgroup\left\|#1\aftergroup\egroup\right\|}
\newcommand{\kl}[1]{\mathopen{}\mathclose\bgroup\left(#1\aftergroup\egroup\right)}
\newcommand{\klg}[1]{\mathopen{}\mathclose\bgroup\left\{#1\aftergroup\egroup\right\}}
\newcommand{\kle}[1]{\mathopen{}\mathclose\bgroup\left[#1\aftergroup\egroup\right]}
\newcommand{\kls}[1]{\mathopen{}\mathclose\bgroup\left\langle#1\aftergroup\egroup\right\rangle}
\newcommand{\floor}[1]{\mathopen{}\mathclose\bgroup\left\lfloor#1\aftergroup\egroup\right\rfloor}
\newcommand{\ceil}[1]{\mathopen{}\mathclose\bgroup\left\lceil#1\aftergroup\egroup\right\rceil}
\newcommand{\leftopen}[1]{\mathopen{}\mathclose\bgroup\left(#1\aftergroup\egroup\right]}

\newcommand{\defi}{\mathrel{\mathop:}=}
\newcommand{\ifed}{=\mathrel{\mathop:}}

\DeclareMathOperator{\sgn}{sgn}

\newcommand{\refrel}[2]{\stackrel{\text{\ref{#1}}}{#2}}
\newcommand{\dblrefrel}[3]{\stackrel{\text{\ref{#1},\ref{#2}}}{#3}}

\newcommand{\Epi}{\affiliation{Department of Epileptology, University of Bonn, Sigmund-Freud-Stra{\ss}e~25, 53105~Bonn, Germany}}
\newcommand{\HISKP}{\affiliation{Helmholtz Institute for Radiation and Nuclear Physics, University of Bonn, Nussallee~14--16, 53115~Bonn, Germany}}
\newcommand{\IZKS}{\affiliation {Interdisciplinary Center for Complex Systems, University of Bonn, Br\"uhler Stra\ss{}e~7, 53175~Bonn, Germany}}

\begin{document}

\title{A highly specific test for periodicity}

\author{Gerrit Ansmann}
\Epi \HISKP \IZKS

\begin{abstract}
We present a method that allows to distinguish between nearly periodic and strictly periodic time series.
To this purpose, we employ a conservative criterion for periodicity, namely that the time series can be interpolated by a periodic function whose local extrema are also present in the time series.
Our method is intended for the analysis of time series generated by deterministic time-continuous dynamical systems, where it can help telling periodic dynamics from chaotic or transient ones.
We empirically investigate our method's performance and compare it to an approach based on marker events (or Poincar\'e sections).
We demonstrate that our method is capable of detecting small deviations from periodicity and outperforms the marker-event-based approach in typical situations.
Our method requires no adjustment of parameters to the individual time series, yields the period length with a precision that exceeds the sampling rate, and its runtime grows asymptotically linear with the length of the time series.
\end{abstract}

\maketitle

\begin{quotation}
	Classifying the dynamics of a system plays an important role in its understanding.
	One of the main classes of dynamics are periodic ones and as a consequence deciding about periodicity is important in various scientific fields.
	While several methods for this exist, there is a shortage of such that can discriminate between nearly periodic and strictly periodic dynamics, as such details are usually lost under experimental conditions (at which these methods are aimed).
	However, the latter does not hold for simulated deterministic systems, which are often used as models to improve our understanding of real systems.
	With such systems in mind, we propose an efficient periodicity test for time series that is capable of detecting small deviations from periodicity.
	We demonstrate our method's capabilities and show that it outperforms existing methods in typical situations.
	While our approach is not aimed at analyzing experimental time series, aspects of it may enhance existing or inspire new techniques for this purpose.
\end{quotation}

\section{Introduction}
Two central distinctions in the classification of dynamical systems are that between a chaotic dynamics and a regular one as well as that between a transient dynamics and a stable one~\cite{Guckenheimer1983, Haykin1983, Hale1991, Strogatz1994, Ott2002, Kantz2003, Gottwald2014}.
As all periodic dynamics are neither chaotic nor transient, being able to decide that a dynamics is periodic is hence a valuable asset when analyzing dynamical systems.
We here consider this issue for the case of simulated time-continuous deterministic systems.

A plethora of methods to identify periodicities in time series have been proposed in the past, which can roughly be divided into four approaches:
\emph{Periodogram\nobreakdash-} or \emph{autocorrelation-based} methods \cite{Burki1978, Siegel1980, Scargle1982, Vlachos2005} search and evaluate local maxima in the frequency domain or in the autocorrelation function.
\emph{Epoch-folding} techniques \cite{Heck1985, Davies1990, Cincotta1999, Larsson1996} are based on sorting the time series' values into bins according to their phase for a presumed period length and finding the optimal period length according to some test statistics.
Recently, approaches have been proposed that search for repeating \emph{patterns}, i.e., subsequences of symbols \cite{Han1999, Ergun2004, Elfeky2005}.
Finally, Ref.~\onlinecite{Rosenblum1995} proposed to search for local minima of the \emph{mean fluctuation function}.

With exception for the latter, neither of these methods aims at detecting small deviations from periodicity, because under experimental conditions, for which these methods were designed, such details can rarely be captured due to noise and other confounding factors.
For example, there may be no indicative difference between the Fourier spectra of chaotic and similar periodic signals~\cite{Rosenblum1995}.
In contrast, time series from simulated systems may reflect small deviations from periodicity, which can be of considerable interest in their analysis.
Moreover, unless the period length is a multiple of the sampling interval, all of the above methods either do not provide a straightforward way to decide whether the time series is periodic \cite{Rosenblum1995, Vlachos2005, Han1999, Ergun2004, Elfeky2005}, e.g., via a critical value of some test statistics,  or test against the null hypothesis that the time series is random \cite{Burki1978, Cincotta1999, Siegel1980, Scargle1982, Heck1985, Davies1990, Larsson1996}, and thus do not distinguish between regular and chaotic or stable and transient time series.

Using embedding techniques~\cite{Takens1981, Marwan2007} or time series of an observable and its temporal derivative~\cite{Brazier1994, Freire2001}, if available, one can reconstruct the dynamics's phase space.
Its structure and other properties may then form the basis of several kinds of analyses \cite{Brazier1994, Freire2001, Kantz2003, Zou2010} --~some of which require a careful selection of parameters and pose a lot of pitfalls~--, but only in case of sinusoidal time series is there a criterion for periodicity~\cite{Brazier1994}.

If the dynamics of a system is known, methods that go beyond the analysis of observables become available:
The theory of dynamical systems allows for an analytical classification, e.g., by analyzing the stable and unstable manifolds of fix points and separatrices between basins of attraction \cite{Hale1991, Ott2002}.
Another approach (often used for continuation problems) is finding a piecewise polynomial approximation of a periodic orbit by numerically solving the respective nonlinear equation system~\cite{Krauskopf2007}.
These approaches may, however, become unfeasible for complicated and high-dimensional dynamics.
A method that does not suffer from these problems is computing the sign of the largest Lyapunov exponent~\cite{Benettin1980}.
Here the difficulty arises that a zero sign (signifying a regular dynamics) can only form the null hypothesis and can thus only be rejected and not accepted.

We here propose a method that allows to decide whether a time series complies with a comparably conservative periodicity criterion and that is thus capable of detecting small deviations from periodicity, such as a slowly rising amplitude or period length.
We first motivate and define this criterion (Sec.~\ref{criterion}) and then describe our approach to tell whether it is fulfilled for a given period length (Sec.~\ref{singleperiod}), for any period length from a small interval (Sec.~\ref{interval}) and finally for any reasonable period length (Sec.~\ref{anylength}).
We then extend our method to account for small, bounded errors, such as numerical ones (Sec.~\ref{errors}).
Using archetypal noise-free (Sec.~\ref{archetypes}) and noise-contaminated (Sec.~\ref{noisy}) test cases, we investigate our method's performance and compare it to an alternative approach based on marker events.
Finally, we demonstrate that our method can help classifying the dynamics of simulated dynamical systems (Sec.~\ref{ode}) and draw our conclusions (Sec.~\ref{conclusions}).

\section{Method}\label{method}

Some general remarks on notation:
Lowercase italic letters denote integers; fraktur denotes fractions; Greek letters denote real numbers.
Uppercase letters denote respectively valued functions.
Blackboard-bold letters denote sets.
The binary modulo operator ranges between addition and multiplication in the order of operations.

The source code of an implementation of our method in~C as well as of a standalone program and a Python module building upon it is freely available~\cite{Source}.

\subsection{General approach and criterion}\label{criterion}

\newcommand{\urbild}{\ensuremath{\klg{0,\ldots,n-1}}\xspace}
\newcommand{\Fn}{\ensuremath{\tilde{\mathbb{F}}_n}\xspace}
\newcommand{\taufoldation}[1][]{\(\tau\)\nobreakdash-foldation#1\xspace}
\newcommand{\tauperiodic}[1][]{\(\tau\)\nobreakdash-periodic#1\xspace}

\begin{figure*}
	\includegraphics[width=\linewidth]{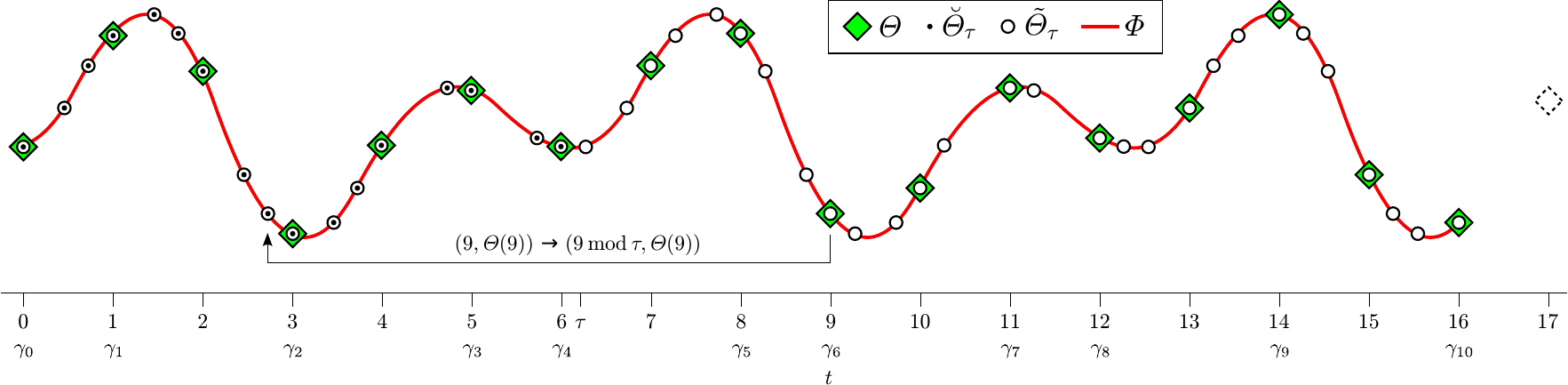}
	\caption{Illustration of several concepts on the example of a time series~\(\Theta\) (green squares) defined on \(\klg{0, \ldots, 16}\) using a period length of \(\tau \approx 6.2\) for folding.
	\(\Phi\)~(red line) is a \tauperiodic extension of \(\Theta\).
	Except for an undetectable final extremum (see below), all of \(\Phi\)'s local extrema are captured by~\(\Theta\) and thus \(\Theta\)~complies with a period length of~\(\tau\) and meets the criteria of Def.~\ref{D:periodic}, respectively.
	The arrow illustrates the process of folding (\(\kl{t, \Theta\kl{t}} \rightarrow \kl{t \bmod \tau, \Theta\kl{t}}\)) for \(t=9\).
	The \taufoldation~\(\breve{\Theta}_\tau\) is marked by small dots; \(\tilde{\Theta}\) is marked by circles.
	\(\gamma_0, \ldots, \gamma_{10}\) indicate the values of one longest sequence \(\gamma\) for counting the extrema of~\(\Theta\) via Def.~\ref{D:extremum}.
	\(\Phi\)~and~\(\tilde{\Theta}_\tau\) have an undetectable final extremum at \(t \approx 15.7\), which could be captured by \(\Theta\) if it lasted until \(t=17\) and contained, e.g., the point marked by the dashed square.
	\(\Phi\)~and~\(\tilde{\Theta}_\tau\) have no undetectable initial extrema.
	}
	\label{fig:concepts}
\end{figure*}

\begin{figure}
	\includegraphics{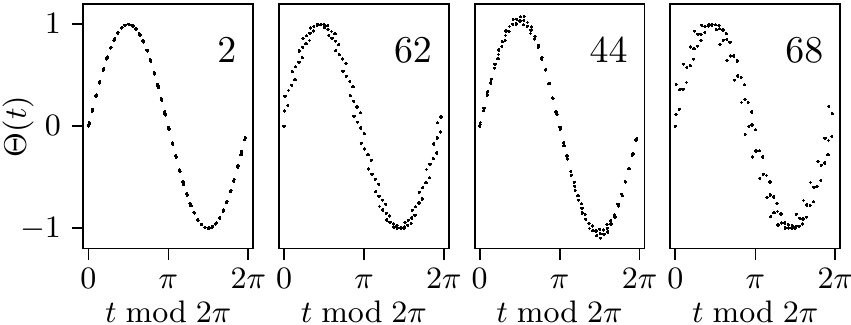}
	\caption{Approximations of one presumed period of length~\(\tau=2\pi\) (\taufoldation) from exemplary time series \(\Theta\) of length \(n=100\).
	From left to right:
	\(\Theta\kl{t} = \sin\kl{t}\);
	\(\Theta\kl{t} = \sin\kl{1.003\,t}\);
	\(\Theta\kl{t} = \kl{1+0.001\,t} \sin\kl{t}\);
	\(\Theta\kl{t} = \sin\kl{\kl{1+0.00005\,t}\,t}\).
	The numbers in the top right indicate the approximation's number of local extrema.
	}
	\label{fig:Beispiele}
\end{figure}

Before delving into the mathematical details, we give a brief overview and illustration over our approach:
We consider a non-constant time series \(\Theta: \urbild \rightarrow \mathbb{R}\) of length~\(n \geq 3\), assuming without loss of generality a sampling rate of~\(1\).
We want to decide about the existence of an extension~\(\Phi\) of~\(\Theta\) to the interval \(\kle{0,n-1}\) that is periodic with a period length~\(\tau\) (which we will initially consider to be fixed) and all of whose local extrema are captured by~\(\Theta\).
The latter criterion serves to avoid strongly oscillating solutions for~\(\Phi\), which could almost always be found.
The function~\(\Phi\) can be understood as the signal underlying~\(\Theta\) (see Fig.~\ref{fig:concepts}).

Following the ansatz of epoch-folding techniques \cite{Heck1985}, we fold all periods into one by mapping \(\kl{t, \Theta\kl{t}}\) to \(\kl{t \bmod \tau, \Theta\kl{t}}\), thus obtaining an approximation of a single presumed period of~\(\Phi\), which we refer to as \emph{foldation, \taufoldation,} or~\(\breve{\Theta}_\tau\) in the following (see Fig.~\ref{fig:concepts}).
We then check whether the foldation's number of local extrema complies with that of~\(\Theta\).
As an example, we show in Fig.~\ref{fig:Beispiele} \taufoldation[s] from four time series, all of which have ca. two local extrema per~\(\tau\) time units.
We observe that small deviations of the period length or from periodicity suffice to drastically increase the foldation's number of local extrema.

As the exact definition of an individual local extremum is not relevant for our method, we only define \textbf{the number} of local extrema of some function~\(\Xi\).
To this purpose we employ the longest zigzagging sequence with points from \(\Xi\)'s graph (\(\kl{\gamma_0, \Xi\kl{\gamma_0}}, \ldots, \kl{\gamma_{k+1}, \Xi\kl{\gamma_{k+1}}}\), see Fig.~\ref{fig:concepts}):

\begin{definition} \label{D:extremum}
	We define \emph{the number of local extrema}~\(E\kl{\Xi}\) of a function \(\Xi : \mathbb{D} \rightarrow \mathbb{R}\) with \(\mathbb{D} \subset \mathbb{R}\) as the largest~\(k\) such that there is  an increasing sequence \(\gamma_0, \ldots, \gamma_{k+1} \in \mathbb{D}\) for which \(0 \neq \sgn\kl{\Xi\kl{\gamma_i}-\Xi\kl{\gamma_{i-1}}} \neq \sgn\kl{\Xi\kl{\gamma_{i+1}}-\Xi\kl{\gamma_{i}}} \neq 0\) for all \(i \in \klg{1, \ldots, k}\).
\end{definition}

Next, we define a correction term for extrema of some function~\(\Xi\) that \(\Theta\)~does not capture, but could capture if it began earlier or lasted longer, i.e., extrema that \(\Theta\)~fails to capture due to its finiteness and not due to \(\Xi\)~strongly oscillating (also see Fig.~\ref{fig:concepts}).

\begin{definition} \label{D:correction}
	We say that \(\Xi : \mathbb{D} \rightarrow \mathbb{R}\) with \(\min\kl{\mathbb{D}} = 0\), \(\max\kl{\mathbb{D}} = n-1\), and \(\urbild \subset \mathbb{D}\) has \emph{undetectable initial extrema} if there is a \(\lambda \in \mathbb{D} \cap \kle{0, 1}\) such that \(\sgn\kl{\Xi\kl{1}-\Xi\kl{0}} \neq \sgn \kl{\Xi\kl{\lambda} - \Xi\kl{0}} \neq 0\).
	Analogously, we define \emph{undetectable final extrema.}
	Finally, to correct for undetectable extrema, we define \(C\kl{\Xi}\) to be~\(2\) if \(\Xi\)~has both undetectable initial and final extrema, \(0\)~if it has neither, and \(1\)~otherwise.
\end{definition}

Finally, we define the property, we are testing for:
\begin{definition} \label{D:periodic}
	We say that a time series \(\Theta: \urbild \rightarrow \mathbb{R}\) \emph{complies with a period length~\(\tau\)} iff an extension \(\Phi : \kle{0,n-1} \rightarrow \mathbb{R}\) of \(\Theta\) exists such that \(\Phi\kl{t} = \Phi\kl{t \bmod \tau} ~\forall t \in \kle{0,n-1}\) and \(E\kl{\Phi} - C\kl{\Phi} = E\kl{\Theta}\).
\end{definition}
\noindent From this definition, it follows that \(\Theta\)~complies with every period length larger than \(n-1\).
Also, \(\Theta\)~complies with no period length shorter than~\(2\), if \(n\)~is sufficiently large.

\subsection{Deciding about periodicity with a given period length}\label{singleperiod}

The following theorem (see Appendix~\ref{proofT1} for a proof) allows us to decide as to whether \(\Theta\)~complies with a period length~\(\tau\) by counting the local extrema of an approximation \(\tilde{\Theta}_\tau\) of all of~\(\Phi\) from a foldation (see Fig.~\ref{fig:concepts}).

\begin{theorem} \label{T:main}
	For a given~\(\tau\), let \(\tilde{\mathbb{D}}_\tau \defi \klg{\alpha \in \kle{0,n-1} \mid \exists i \in \urbild: i \bmod \tau = \alpha \bmod \tau}\).
	For \(\alpha \in \tilde{\mathbb{D}}_\tau\), let \(l_\alpha \in \urbild\) be defined such that \(l_\alpha \bmod \tau = \alpha \bmod \tau\).
	Finally define \(\tilde{\Theta}_\tau: \tilde{\mathbb{D}}_\tau \rightarrow \mathbb{R}\) via \(\tilde{\Theta}_\tau \kl{\alpha} \defi \Theta\kl{l_\alpha}\).
	Then \(\Theta\) complies with a period length~\(\tau\), if and only if \(\tilde{\Theta}_\tau\) is well-defined and \(E(\tilde{\Theta}_\tau) - C(\tilde{\Theta}_\tau)= E\kl{\Theta}\).
\end{theorem}

\(\tilde{\Theta}_\tau\) can only be ill-defined if there are \(i,j \in \urbild\) with \(i \neq j\) and \(i \bmod \tau = j \bmod \tau\), which in turn happens, iff \(\tau \in \Fn\), where \(\Fn \subset \mathbb{Q}\) is the set of fractions with a numerator smaller than~\(n\), i.e., the element-wise inverse of the \(n\)-th Farey sequence~\(\mathbb{F}_n\) (ignoring the latter's restriction to \(\kle{0,1}\)).

Instead of determining \(E(\tilde{\Theta}_\tau)\), it often suffices to regard the number of extrema of the foldation~\(\breve{\Theta}_\tau\), --~i.e., the restriction of~\(\tilde{\Theta}_\tau\) to \(\left [ 0,\tau \right ) \cap \tilde{\mathbb{D}}_\tau\)~-- and use it to extrapolate a lower bound for \(E(\tilde{\Theta}_\tau)\), namely:
\begin{equation}\label{eq:extrapolate}
L_\tau\kl{\Theta} \defi E(\breve{\Theta}_\tau) \floor{\frac{n}{\tau}} + B(\breve{\Theta}_\tau) \kl{\floor{\frac{n}{\tau}} -1} \leq E(\tilde{\Theta}_\tau),
\end{equation}
where \(B(\breve{\Theta}_\tau) \in \klg{0,1,2}\) is the number of extrema that are so close to multiples of \(\tau\) that they are not accounted for by \(E(\breve{\Theta}_\tau)\).
Moreover, if \(\Theta\) does not comply with a period length~\(\tau\), \(L_\tau\kl{\Theta}\) is often much larger than \(E(\Theta)\) (see also Fig.~\ref{fig:Beispiele}).
In this case it suffices to regard a few values of the foldation to reject that \(\Theta\)~complies with a period length~\(\tau\).

To determine the values of \(\breve{\Theta}_\tau\), functions with the following property are useful:

\begin{definition}
	We say that a function \(I\) \emph{sorts the first~\(b\) integers modulo~\(\tau\)} iff it maps \(\klg{0,\ldots,b-1}\) to a permutation of itself and
	\[I\kl{i} \bmod \tau \leq I\kl{i+1} \bmod \tau \quad \forall i \in \klg{0, \ldots, b		-2}.\]
\end{definition}

\noindent The function \(I\) that sorts the first~\(n\) integers modulo~\(\tau\) is unique and thus strictly monotonically increasing, if and only if \(\tau \notin \Fn\).
In this case, the values of \(\Theta \circ I\) are identical to the values of \(\breve{\Theta}\) and assumed in the same order and, in particular, \(E\kl{\Theta \circ I} = E(\breve{\Theta}_\tau)\).
Moreover, we only need \(\klg{\Theta\kl{I\kl{0}}, \ldots, \Theta\kl{I\kl{n-1}}}\), \(n\), and \(I^{-1}\kl{n-1}\) to determine the value sequence of \(\tilde{\Theta}_\tau\) and thus \(E(\tilde{\Theta}_\tau)\).
With additional knowledge of \(I^{-1}\kl{1}\) and \(I^{-1}\kl{n-2}\) we can also determine~\(C(\tilde{\Theta}_\tau)\).
As knowing \(E(\tilde{\Theta}_\tau)\) and~\(C(\tilde{\Theta}_\tau)\) (as well as \(E\kl{\Theta}\)) suffices to decide whether \(\Theta\)~complies with a period length~\(\tau\) (Th.~\ref{T:main}), the latter only depends on~\(I\), with no further explicit dependence on~\(\tau\).

\subsection{Deciding about periodicity with period lengths within small intervals}\label{interval}

If we know \(\tau\), we can easily find a function that sorts the first~\(n\) integers modulo~\(\tau\) by sorting.
For ranges of \(\tau\), we require the following theorem (see Appendix~\ref{proofT2} for a proof):

\begin{theorem}\label{T:sorting}
	Let \(1<\tau<n\) and let \(\frac{p}{q}\) be the largest  and~\(\frac{r}{s}\) be the smallest reduced fraction from \Fn such that \(\frac{p}{q} < \tau \leq \frac{r}{s}\).
	Define \(I_{n,\tau}\kl{i} \defi ir \bmod \kl{p+r}\) for \(i\in \klg{0,\ldots,p+r}\).
	Then \(I_{n,\tau}\) sorts the first~\(p+r\) integers modulo~\(\tau\).
	Moreover, \(I_{n,\tau}\) increases strictly monotonically on \urbild if \(\tau<\frac{r}{s}\).
\end{theorem}

It follows that, for two successive elements \(\mathfrak{a}\) and~\(\mathfrak{b}\) of \Fn with \(\mathfrak{a}<\mathfrak{b}\), \(I_{n, \tau}\) is the same for all \(\tau \in \leftopen{\mathfrak{a},\mathfrak{b}}\).
Therefore, if \(\Theta\)~complies with one period length \(\tau\) in \(\kl{\mathfrak{a}, \mathfrak{b}}\), it complies with all period lengths in that interval.
Moreover, if \(\Theta\) complies with a period length~\(\mathfrak{b}\), it also complies with all period lengths in \(\kl{\mathfrak{a}, \mathfrak{b}}\) and thus it suffices to investigate open intervals of adjacent elements of~\Fn to decide about the periodicity of~\(\Theta\).
As a consequence, all we need to do to check whether \(\Theta\)~complies with period lengths in \(\kl{\mathfrak{a}, \mathfrak{b}}\) is to count the extrema of one foldation.

In another consequence, Th.~\ref{T:sorting} allows us to iterate over the values of \(I_{n,\tau}\) and thus of \(\Theta \circ I_{n,\tau}\) without explicitly calculating~\(I_{n, \tau}\kl{i}\) for all \(i \in \urbild\).
As we mentioned in the previous subsection, a few such iterations often suffice to reject that \(\Theta\)~complies with a period~\(\tau\).

\subsection{Deciding about periodicity with an unknown period length}\label{anylength}

If we have little prior knowledge or constraints on the value of a possible period length~\(\tau\), testing every possible interval of adjacent elements of \Fn will usually be unfeasible because \(\big|\Fn\big| = \mathcal{O}\big(n^2\big)\) as \(n \rightarrow \infty\) \cite{*[{Ex.~8.4.i in }][] Vardy1991}.
To avoid this, we can make use of the following:
First, if \(\Theta\)~complies with a period length~\(\tau\), so does any segment of~\(\Theta\).
Second, if the length~\(\breve{n}\) of such a segment is sufficiently smaller than~\(n\), the interval of adjacent elements of~\(\tilde{\mathbb{F}}_{\breve{n}}\) in which a given~\(\tau\) lies is larger than the analogous interval for~\Fn.
Third, for each two adjacent elements on some level of the Stern--Brocot tree\footnote{The Stern--Brocot~\cite{Graham1989} tree is a tree spanning all reduced fractions, which can be recursively defined as follows:
Let \(\mathfrak{s}^i_1, \ldots, \mathfrak{s}^i_{2^i+1}\) denote the fractions on the \(i\)-th level of the tree.
Then \(\mathfrak{s}^0_1 \defi \tfrac{-1}{0}\),
\(\mathfrak{s}^0_2 \defi \tfrac{1}{0}\),
\(\mathfrak{s}^i_{2j} \defi  \mathfrak{s}^{i-1}_j\),
and \(\mathfrak{s}^i_{2j+1} \defi \mathfrak{M}\big(\mathfrak{s}^{i-1}_j, \mathfrak{s}^{i-1}_{j+1}\big)\),
where \(\mathfrak{M}\) denotes the \emph{mediant:} \(\mathfrak{M}\big(\tfrac{x}{y},\tfrac{z}{w}\big) \defi \tfrac{x+z}{y+w}\).
Each level contains fractions in ascending order, i.e, \(\mathfrak{s}^i_1 < \ldots < \mathfrak{s}^i_{2^i+1}\).}\nocite{Graham1989}, there is some~\(o\) such that they are adjacent elements of~\(\tilde{\mathbb{F}}_o\).
Combining these three facts, we can make a nested-interval search on the Stern--Brocot tree for a~\(\tau\) with which \(\Theta\) complies, as follows:

\begin{figure}
	\includegraphics[width=\linewidth]{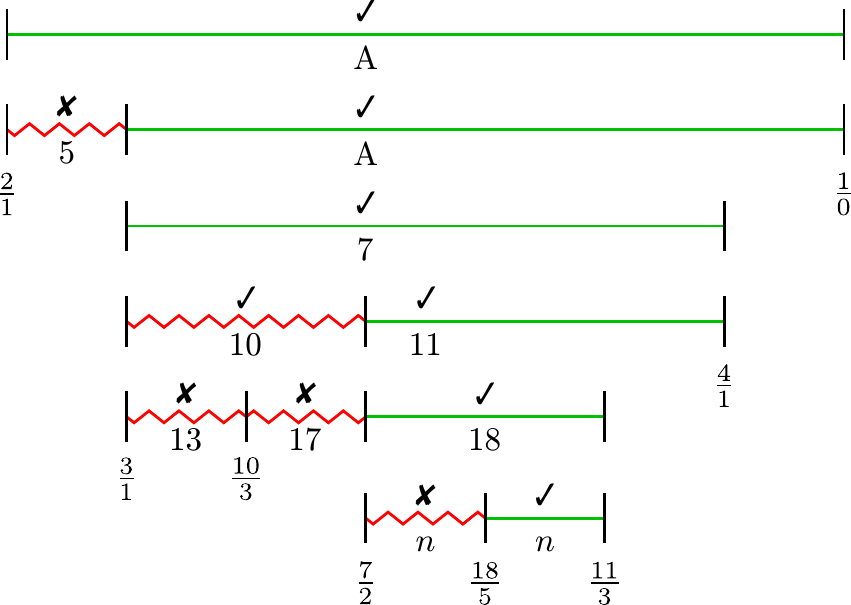}
	\caption{Illustration of a possible realization of the nested-interval search for a time series with \(18 < n \leq 25\) that complies with period lengths in \(\leftopen{\tfrac{18}{5},\tfrac{11}{3}}\).
	Each level corresponds to one level of the Stern--Brocot tree.
	The numbers below the intervals indicate the length of the segment checked in step \ref{finaltest} or~\ref{quicktest} of the procedure for the respective interval, i.e., \(n\) or~\(\breve{n}\), respectively.
	``A''~indicates an automatic pass due to \(\tfrac{1}{0}\) being the interval's right border.
	Checkmarks indicate whether this check was positive or negative; only displayed intervals were checked at all.
	Green, straight lines indicate that the test found \(\Theta\)~to comply with some period length in the respective interval; red, zigzagged lines indicate that it did not.
	The check on the interval \(\leftopen{\tfrac{3}{1},\tfrac{7}{2}}\) is an example for a false positive result in step~\ref{quicktest}.
	Intervals are not displayed to scale.
	\label{fig:Algorithmus}
}
\end{figure}

\textbf{Method:} To test whether \(\Theta\)~complies with some period length in \(\leftopen{\tfrac{p}{q}, \tfrac{r}{s}}\) with \(\tfrac{p}{q}\) and \(\tfrac{r}{s}\) being successive elements of \(\tilde{\mathbb{F}}_m\) for some \(m \leq n\):
\renewcommand{\theenumi}{\arabic{enumi}}
\renewcommand{\labelenumi}{(\theenumi)}
\begin{enumerate}
	\item \label{finaltest} If \(p+r \geq n\): Check whether \(\Theta\)~complies with period lengths from \(\kl{\tfrac{p}{q}, \tfrac{r}{s}}\) (using Th.~\ref{T:main}).\\
	If yes, the test is positive.
	If no, the test is negative.
	\item \label{quicktest} If \(p+r < n\), check whether some segment of~\(\Theta\) of length \(\breve{n} \defi p+r\) complies with period lengths in \(\kl{\tfrac{p}{q}, \tfrac{r}{s}}\) (using Th.~\ref{T:main}).
	If \(\tfrac{r}{s}=\tfrac{1}{0}\), this check can be skipped as it is always positive.\\
	If no, the test is negative.
	If yes, continue.
	\item Use this method to test whether \(\Theta\)~complies with some period length in \(\leftopen{\tfrac{p}{q}, \tfrac{p+r}{q+s}}\).\\
	If yes, the test is positive.
	If no, continue.
	\item Use this method to test whether \(\Theta\)~complies with some period length in \(\leftopen{\tfrac{p+r}{q+s}, \tfrac{r}{s}}\).\\
	If yes, the test is positive.
	If no, the test is negative.
\end{enumerate}
Without any prior knowledge or constraints on the value of a possible period length~\(\tau\), the above test can be applied to \(\leftopen{\tfrac{2}{1}, \tfrac{1}{0}}\), i.e., with \(p=2\), \(q=r=1\), and \(s=0\), neglecting all intervals whose smallest value exceeds some \(\tau_\text{max} \leq n-1\).
A realization of this procedure is illustrated in Fig.~\ref{fig:Algorithmus}.

Some remarks on the implementation:
\begin{itemize}
	\item If the check in step~\ref{quicktest} yielded a false positive result (such as for the interval \(\leftopen{\tfrac{3}{1},\tfrac{7}{2}}\) in Fig.~\ref{fig:Algorithmus}), this would not affect the total outcome of the test, as all subintervals will be tested again (in step~\ref{finaltest} or~\ref{quicktest}) at a higher recursion level.
	(Some segment of~\(\Theta\) may comply with a period length~\(\tau\) while \(\Theta\) itself does not, anyway.)
	This can be employed to make an implementation more effective by using the number of extrema of the foldation of the respective segment to extrapolate the number of extrema of~\(\tilde{\Theta}_\tau\) and comparing it to \(E\kl{\Theta}\) (see Eq.~\ref{eq:extrapolate}).
	This way, calculating or storing the number of local extrema of each segment can be avoided.
	\item We empirically found that the runtime was increased if we used a random segment of length \(p+r\) instead of the first one in step~\ref{quicktest}.
	\item As the interval \(\leftopen{\tfrac{m}{1},\tfrac{m+1}{1}}\) for some \(m\leq n\) resides at the \(\kl{m-1}\)-st level of the employed branch of the Stern--Brocot tree, a fully recursive implementation can result in a problematically high recursion depth for large \(n\) and \(\tau_\text{max}\).
	This can be avoided by applying the above test to the intervals \(\leftopen{\tfrac{2}{1},\tfrac{3}{1}}, \leftopen{\tfrac{3}{1},\tfrac{4}{1}}, \ldots, \leftopen{\tfrac{\ceil{\tau_\text{max}}-1}{1},\tfrac{\ceil{\tau_\text{max}}}{1}}\) in a non-recursive manner.
	\item The above method can also be used to find the shortest period length that \(\Theta\)~complies with.
	For obtaining an estimate with an error margin for the shortest period length of the underlying process, one usually would want to know the first maximal interval such that \(\Theta\)~complies with every period length in that interval.
	To avoid underestimating the margin, one needs to consider that if \(\mathfrak{c}\) is the successor of the successor of \(\mathfrak{a}\) in~\Fn, \(\Theta\)~may comply with all period lengths in \(\leftopen{\mathfrak{a},\mathfrak{c}}\).
	Therefore, after finding an interval such that \(\Theta\)~complies with every period length in that interval, the next interval of adjacent elements of \Fn needs to be checked as well.
\end{itemize}

In the following, we consider and use an implementation that takes all of the above into account.
The asymptotic runtime of our method scales linear with~\(n\), if \(\Theta\) complies with a period length \(\tau \ll n\) and quadratic with \(\tau_\text{max}\) otherwise (see Appendix~\ref{runtime}).

\subsection{Accounting for small errors}\label{errors}
In this subsection, we describe a simple expansion of our method that is capable of allowing for errors of the time series~\(\Theta\) that are bounded and small in comparison to \(\Theta\)'s features of interest.
Such errors could originate from small numerical inaccuracies, e.g., of a solver for differential equations, or small measurement errors.
We handle such errors by treating two values of~\(\Theta\) that are less apart than a given \emph{error allowance}~\(\sigma\) as identical for the purpose of determining local extrema.
More precisely, we weaken our criterion (Def.~\ref{D:periodic}) by changing Defs. \ref{D:extremum} and~\ref{D:correction} as follows:
\begin{itemize}
	\item In Def.~\ref{D:extremum}, we additionally require that \(\abs{\Xi\kl{\gamma_{i+1}}-\Xi\kl{\gamma_i}} > \sigma\) for all \(i \in \klg{0, \ldots, k}\).
	\item In Def.~\ref{D:correction}, if \(i\) and~\(j\) with \(i<j\) are the smallest elements of \urbild such that \(\abs{\Xi\kl{j}-\Xi\kl{i}}>\sigma\), we say that \(\Xi\) has \emph{undetectable initial extrema} if there exist \(\lambda, \mu \in \mathbb{D} \cap \kle{0,j}\) such that \(\abs{\mu-\lambda}>\sigma\) and \(\sgn\kl{\mu - \lambda} \neq \sgn\kl{\Xi\kl{j}-\Xi\kl{i}}\).
\end{itemize}
For \(\sigma=0\), the adapted Def.~\ref{D:extremum} is equivalent to the original one, and so is the adapted Def.~\ref{D:correction}, unless \(\Theta\kl{0}=\Theta\kl{1}\) or \(\Theta\kl{n-2}=\Theta\kl{n-1}\).
The remainder of Sections \ref{criterion} to \ref{anylength} as well as the appendices analogously apply to the expanded method.

\section{Application}\label{application}

\subsection{Data with known periodicities or aperiodicities}\label{archetypes}

\newcommand{\ttau}{\ensuremath{\tilde{\tau}}\xspace}

\begin{figure}
	\includegraphics{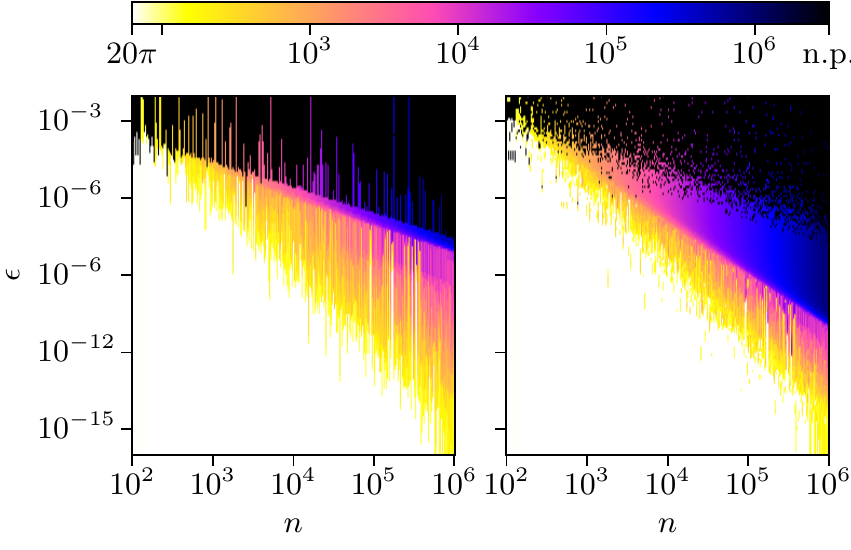}
	\caption{Smallest period length~\(\tau\leq n-1\) with which the aperiodic time series \(\Upsilon_{\epsilon n \ttau}\) (left) or \(\Psi_{\epsilon n \ttau}\) (right) complies for \(\ttau = 20\pi\).
	Black points (n.p.) indicate that no such period length exists.}
	\label{fig:n-eps-found_tau}
\end{figure}

To investigate the performance of our test, respectively, we employed the following families of time series that deviate from periodic ones to an adjustable extent:
\begin{align}
	\begin{split}
	&\Upsilon_{\epsilon n \ttau}: \urbild \rightarrow \mathbb{R}\\
	&\Upsilon_{\epsilon n \ttau}\kl{t} \defi \kl{1+\epsilon \frac{2 \pi t}{\ttau}} \cos\kl{\frac{2 \pi t}{\ttau}},
	\end{split}\\
	\begin{split}
	&\Psi_{\epsilon n \ttau}: \urbild \rightarrow \mathbb{R}\\
	&\Psi_{\epsilon n \ttau}\kl{t} \defi \cos\kl{\frac{\frac{2 \pi t}{\ttau}}{1+\epsilon \frac{2 \pi t}{\ttau}}},
	\end{split}
\end{align}
with \(\epsilon, \ttau \in \mathbb{R}^+\).
The time series~\(\Upsilon_{\epsilon n \ttau}\) features a rising amplitude; the time series \(\Psi_{\epsilon n \ttau}\) features a rising period length.
The parameter~\(\epsilon\) determines how quickly the amplitude or period length, respectively, is rising, i.e., how strongly the time series deviates from a periodic one (at \(\epsilon=0\)), whose period length is~\ttau.

In Fig.~\ref{fig:n-eps-found_tau}, we show the smallest period lengths~\(\tau\) with which these time series comply depending on \(n\) and~\(\epsilon\) for \(\ttau = 20\pi\).
For both families of time series, we find that \(\tau \approx \ttau\) for small \(n\) and~\(\epsilon\), more precisely for \(\epsilon n^3 \lessapprox 10^2\), i.e., for a relative amplitude/frequency change smaller than \(\tfrac{10}{n^2}\).
For higher \(n\) and~\(\epsilon\), the period length~\(\tau\) tendentially increases with \(\epsilon\) and~\(n\), with \(\tau\)~obtaining higher values in general for \(\Psi_{\epsilon n \ttau}\) than for \(\Upsilon_{\epsilon n \ttau}\).
Most time series with high \(n\) and~\(\epsilon\) (roughly: \(\epsilon n > 10^{-2}\), relative amplitude/frequency change bigger than \(10^{-3}\)) do not comply with any period length smaller than \(n-1\) and if they do, \(\tau\)~is close to~\(n\).
For \(\ttau=2\pi\), \(\tau\) exhibits comparable dependencies on \(n\) and~\(\epsilon\), however, with some anomalies due to aliasing effects (not shown).

These results show that our test is capable of detecting small deviations from periodicity.
However, there is a considerable set of parameters for which periodicity with a long period length is detected.
Moreover, even time series with strong deviations from periodicity sometimes comply with period lengths close to~\(n\).
This can be explained by the fact that a time series~\(\Theta\) can easily ``accidentally'' comply with a period length close to its length~\(n\), e.g., for \(\Theta\) to comply with period lengths from \(\kl{n-2,n-1}\), it suffices that \(\Theta\kl{n-1}\) lies between \(\Theta\kl{0}\) and \(\Theta\kl{1}\).
This demonstrates the importance of choosing \(\tau_\text{max}\) properly.

\begin{figure}
	\includegraphics{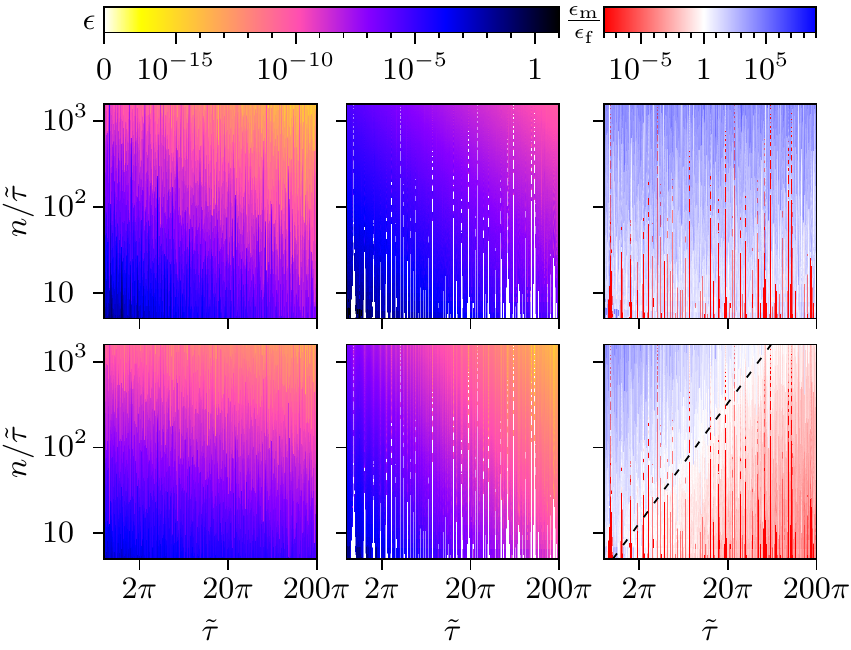}
	\caption{Left column: Smallest \(\epsilon_\text{f}\) such that the aperiodic time series \(\Upsilon_{\epsilon_\text{f} n \ttau}\) (top) or \(\Psi_{\epsilon_\text{f} n \ttau}\) (bottom) does not comply with a period length less than 1.5\,\ttau as per our \textbf{f}oldation-based test.
	Central column: Smallest \(\epsilon_\text{m}\) such that \(\Upsilon_{\epsilon_\text{m} n \ttau}\) (top) or \(\Psi_{\epsilon_\text{m} n \ttau}\) (bottom) are aperiodic as per the \textbf{m}arker-event test.
	Right column: Quotient of \(\epsilon_\text{m}\) and~\(\epsilon_\text{f}\).
	For~\(\Upsilon\), \(\epsilon_\text{m}\)~is larger than~\(\epsilon_\text{f}\) everywhere, except if \(\epsilon_\text{m}=0\) (see top middle panel).
	For~\(\Psi\), \(\epsilon_\text{m}\)~is larger than~\(\epsilon_\text{f}\) in the top left half of the diagram (except if \(\epsilon_\text{m}=0\)) and smaller in the bottom right half.
	The dashed black line in the bottom right panel marks the parameters investigated in Fig.~\ref{fig:compare_for_noise}.}
	\label{fig:n-tau-eps}
\end{figure}

We now compare our test's performance to a test based on Poincar\'e sections or marker events~\cite{Pikovsky2001}, which we refer to as \emph{marker-event test.}
As marker events we employ on the one hand the upward zero crossings of a piecewise linear interpolation of the time series and on the other hand the time series's local maxima.
We consider a time series periodic according to the marker-event test, if neither the distances of subsequent zero crossings nor the amplitudes of local maxima are significantly correlated with time (with a significance level of \(0.05\) as per Kendall's correlation coefficient).
To determine the period length with this test, we use the mean of the distances of subsequent zero crossings, and as its confidence interval, we use twice the standard error.

As a first benchmark, we use the lowest deviation from periodicity \(\epsilon\) for which the test detects \(\Upsilon_{\epsilon n \ttau}\) or \(\Psi_{\epsilon n \ttau}\), respectively to be aperiodic (Fig.~\ref{fig:n-tau-eps}).
We find both tests to be more specific for higher \(n\) and~\ttau in general.
However, the marker-event test detects even the purely sinusoidal \(\Upsilon_{0 n \ttau}=\Psi_{0 n \ttau}\) to be aperiodic in some instances, which is expected given the possibility of type~I errors by the statistical test.
Apart from these cases, our test is generally capable of detecting smaller deviations from periodicity in the form of increasing amplitude~(\(\Upsilon\)) than the marker-event test (blue points in Fig.~\ref{fig:n-tau-eps}, top right).
For changes of the period length~(\(\Psi\)), the marker-event test performs better for high~\(\tau\) and small~\(n\), namely for \(n \lessapprox \ttau^{2.4}\).

\begin{figure}
	\includegraphics{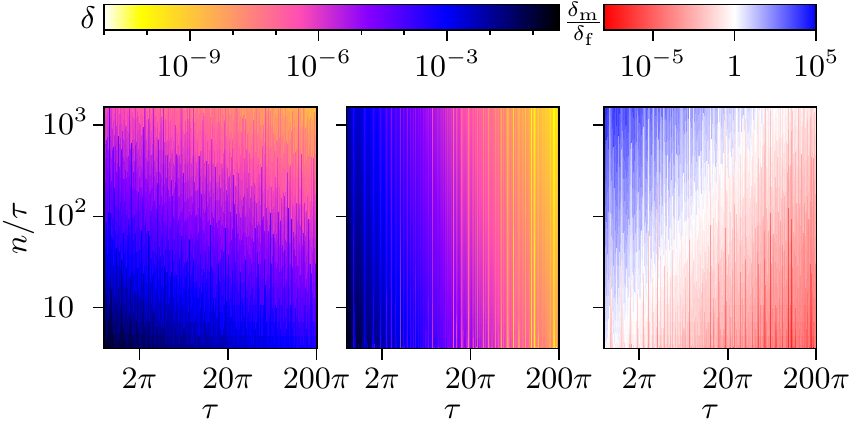}
	\caption{Left: Relative size~\(\delta_\text{f}\) of the confidence interval for the period length for our test for sinusoidal time series with period length~\(\tau\) and \(n\)~sampling points.
	(\(\delta_\text{f} \defi \frac{\mathfrak{a}-\mathfrak{b}}{\tau}\), with \(\kl{\mathfrak{a},\mathfrak{b}}\) being the first maximal interval such that the time series complies with every period length in that interval.)
	Center: relative size~\(\delta_\text{m}\) of the confidence interval for the marker-event test.
	(\(\delta_\text{m} \defi \frac{4 \varsigma}{\tau}\), with \(\varsigma\)~being the standard error of the distances between adjacent upward zero crossings.)
	Right: Quotient between the two.
	\(\delta_\text{m}\) is larger than \(\delta_\text{f}\) in the top left half and smaller in the bottom right half.}
	\label{fig:period_length}
\end{figure}

As a second benchmark, we regarded the error margin of the estimate of the period length for sinusoidal time series (Fig.~\ref{fig:period_length}).
For both tests, we find the margin to decrease with \(\tau\) and~\(n\), however, the latter decrease is small for the marker-event test.
For high~\(\tau\) and low~\(n\), the error margin is higher for our test (blue points in Fig.~\ref{fig:period_length}, right), while it is higher for the marker-event test for low~\(\tau\) and high~\(n\), more precisely for \(n \lessapprox \tau^{2.3}\).
For 5305 of the 90000 time series analyzed for Fig.~\ref{fig:period_length}, the actual period length~(\(\tau\)) did not lie inside the marker-event test's error margin, while it did always lie within the margin for our test.

Our results show that our test outperforms the marker-event test for coarse sampling and a high number of data points as well as for rising amplitudes.
Moreover, it has no false positives and can thus be regarded to be more robust.
Note that the marker-event test used here was tailored to the investigated time series (by assuming one upward zero crossing and one local maximum per period) and to the types of deviations from periodicity (by assuming a rising amplitude or period length), while our test requires no comparable adjustment.
Finally, if \(\tau_\text{max}\) is not chosen too high, our test's asymptotic run-time behavior is better than that of the marker-event test, which is \(\mathcal{O}\kl{n\log\kl{n}}\) due to the correlation coefficient~\cite{Christensen2005}.

\subsection{Noisy data}\label{noisy}

\begin{figure}
	\includegraphics{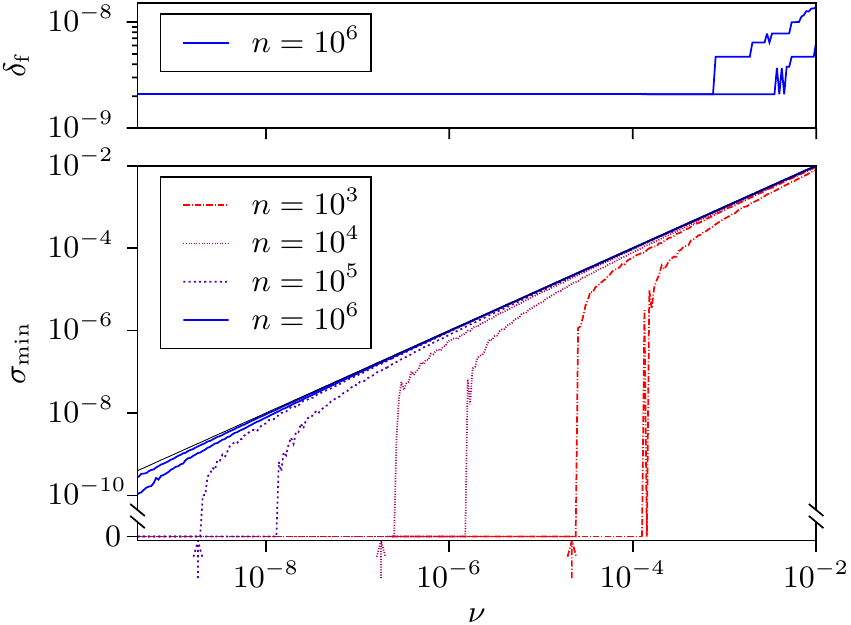}
	\caption{Bottom: Smallest error allowance \(\sigma_\text{min}\) for which sinusoidal time series of length~\(n\) with a period length \(\tau=200\pi\) that were contaminated with white noise from \(\mathcal{U}\kl{\kle{0,\nu}}\) comply with a period length of~\(\tau\).
	The black solid line indicates \(\sigma_\text{min}=\nu\).
	Arrows indicate the third-smallest distance between adjacent values of the \taufoldation.
	Top: Relative size~\(\delta_\text{f}\) of the confidence interval for the period length (cf. Fig.~\ref{fig:period_length}) for \(n=10^6\) and an error allowance \(\sigma=\nu\).
	For \(n=10^5\), a comparable behavior was observed (with \(\delta_\text{f}\)~being generally larger); for \(n=10^4\) and \(n=10^3\), \(\delta_\text{f}\) was constant for \(10^{-10}<\nu<10^{-2}\).
	In both plots, the 5\textsuperscript{th} and 95\textsuperscript{th} percentile over 200 realizations of the noise are shown.}
	\label{fig:noise}
\end{figure}

To investigate the impact of erroneous data on our test, we first apply it to sinusoidal time series that are contaminated with white noise from \(\mathcal{U}\kl{\kle{0,\nu}}\), with \(\mathcal{U}\) denoting the uniform distribution.
In the bottom part of Fig.~\ref{fig:noise}, we show the minimum error allowance \(\sigma_\text{min}\) (see Sec.~\ref{errors}) that needs to be made for such a time series to comply with the correct period length.
We find that the noise does not affect the result up to a certain noise amplitude \(\nu_\text{crit}\).
For up to roughly \(6\nu_\text{crit}\), the noise's impact is strongly realization-dependent.
For higher noise amplitudes, \(\sigma_\text{min}\) is slightly smaller than the noise amplitude~\(\nu\).
We explain these regimes as follows:
For \(\nu \lessapprox \nu_\text{crit}\), the distances between consecutive values of the foldation are larger than \(\nu\) and thus the noise cannot introduce additional local extrema.
This is confirmed by the observation that \(\nu_\text{crit}\) roughly corresponds to the third-smallest such distance (the smallest and second-smallest occur at local extrema, where changing the order of values does not affect the test's outcome; see arrows in Fig.~\ref{fig:noise}).
For \(\nu \gtrapprox 6\nu_\text{crit}\), the probability that the noise did not introduce any additional local extrema to the foldation becomes negligible and thus an error allowance of roughly \(\nu\) is needed for compensation.
Moreover, we find that for small noise levels the accuracy of the identified period length is not affected (see the top part of Fig.~\ref{fig:noise}).
We conclude that \(\sigma=\nu\) is an appropriate choice, given a known noise amplitude~\(\nu\).

\begin{figure}
	\includegraphics{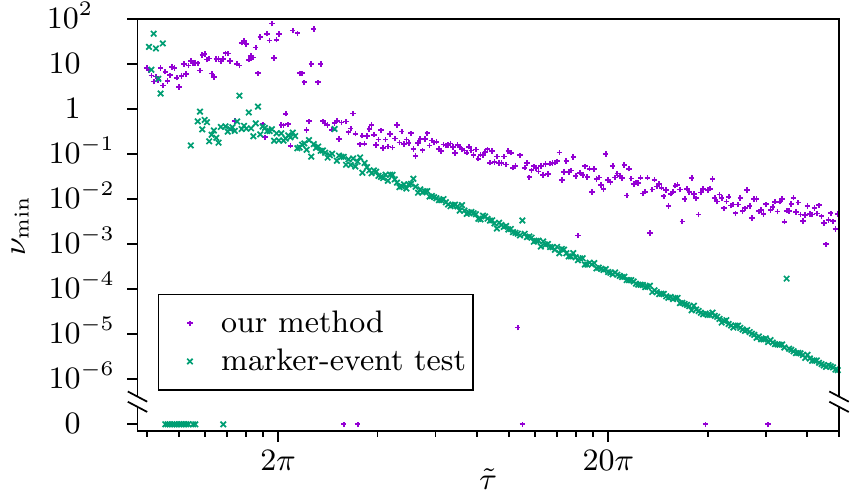}
	\caption{Minimum noise level~\(\nu_\text{min}\) such that \(\Psi_{4.1 \cdot \ttau^{-5.3}, \ttau^{2.4}, \ttau}\) (see text) contaminated with a white noise from \(\mathcal{U}\kl{\kle{-\frac{\nu}{2},\frac{\nu}{2}}}\) is (wrongly) detected to be periodic with a period length of about \(\ttau\) according to our test (with an error allowance~\(\sigma=\nu\)) or the marker-event test, respectively.}
	\label{fig:compare_for_noise}
\end{figure}

To evaluate our test's robustness against noise and to compare it with the marker-event test, we employ \(\Psi_{\epsilon n \ttau}\) contaminated with white noise \(\mathcal{U}\kl{\kle{-\frac{\nu}{2},\frac{\nu}{2}}}\).
To exclude other factors that may influence the relative performance of the tests, we set \(n=\ttau^{2.4}\), for which \(\epsilon_\text{m} \approx \epsilon_\text{f}\) for the uncontaminated time series (see the dashed line in the bottom right of Fig.~\ref{fig:n-tau-eps}).
For these cases we found that \(\epsilon_\text{m} \approx \epsilon_\text{f} \approx \Omega \kl{\ttau} = 4.1 \cdot \ttau^{-5.3}\), with the coefficients of the latter being obtained by a fit.
We chose \(\epsilon = \xi \Omega\kl{\ttau}\) with \(\xi=10\).
In Fig.~\ref{fig:compare_for_noise}, we show the minimum noise level \(\nu_\text{min}\) for which our test (with an error allowance~\(\sigma=\nu\)) or the marker-event test fail to detect the aperiodicity of \(\Psi_{\nu n \ttau}\) for the above conditions.
For our method, we observe \(\nu_\text{min}\) to slightly increase with~\ttau for low~\ttau, being higher than the amplitude of the uncontaminated time series.
The time series that were detected to be periodic were constant with respect to the error allowance~\(\sigma\), i.e., \(\max\kl{\Theta}-\min\kl{\Theta} < \sigma\).
Around \(\ttau \approx 2\pi\), \(\nu_\text{min}\) quickly decreases to about \(0.5\), after which it decreases more slowly with some power law.
We explain these two regimes as follows:
If \(\nu\) (and thus \(\sigma\)) is smaller than some value between \(0.5\) and~\(3\), periodicity is detected when the error allowance prevents the test from detecting the uncontaminated time series's deviations from periodicity.
If, however, \(\nu\) is larger, the noise dominates the original signal and the original signal acts like a contamination (and thus the condition that the errors are small in comparison to the time series' features of interest is not met anymore).
Under an error allowance \(\sigma=\nu\), the noise is constant in the terms of our test and thus periodic.
Thus, periodicity is only detected, when the noise amplitude becomes so high that the uncontaminated time series's influence on the noise becomes negligible.
For the marker-event test, \(\nu_\text{min}\) mostly follows a power law, the main exception being a few cases, in which \(\nu_\text{min}=0\).
We made comparable observations for \(\xi=100\) and \(\xi=1000\).
In general \(\nu_\text{min}\) is higher for our test, which indicates that it is less affected by noise, provided the error allowance~\(\sigma\) can be chosen to match the noise amplitude.

\subsection{Dynamical Systems}\label{ode}

\begin{figure}
	\includegraphics{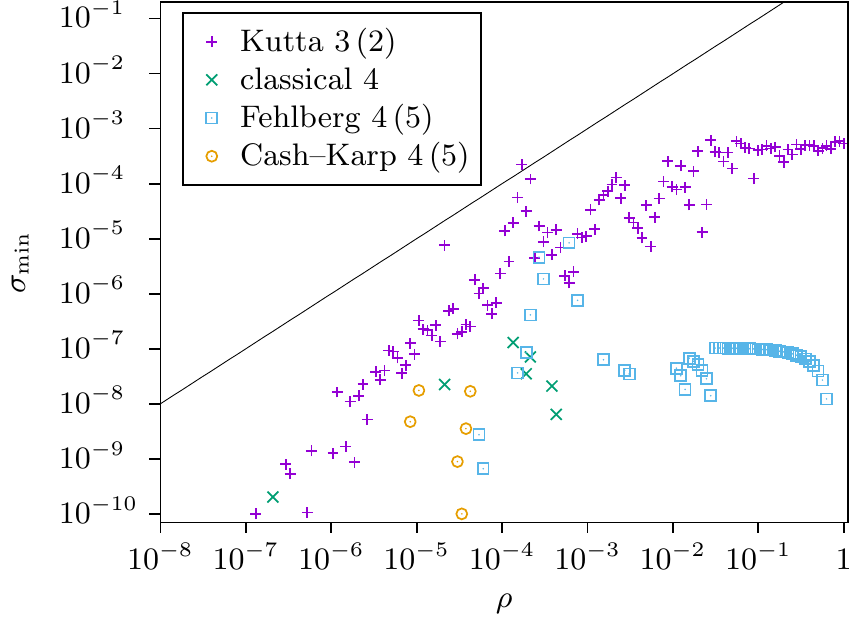}
	\caption{Minimum error allowance~\(\sigma_\text{min}\) such that a time series of length \(10^5\) generated by a model system of two coupled FitzHugh--Nagumo oscillators (see text) integrated with different adaptive Runge--Kutta methods with a maximum relative error~\(\rho\) and a maximum step size of~1 complied with some period length smaller than 500.
	The solid line indicates \(\sigma_\text{min}=\rho\).
	Only cases with \(\sigma_\text{min} \neq 0\) are depicted.
	The number of cases with \(\sigma_\text{min}=0\) were:
	Kutta:            37;
	classical:       153;
	Fehlberg:        112;
	Cash--Karp:      154;
	Dormand--Prince: 160 (all).
	}
	\label{fig:integratortest}
\end{figure}

To evaluate our method's performance on the analysis of dynamical systems, we apply it to a deterministic system of two diffusively coupled FitzHugh--Nagumo oscillators~\cite{FitzHugh1961, Ansmann2013, Karnatak2014}.
We employ a parameter range, in which this system exhibits several regimes of (periodic) mixed-mode oscillations (MMOs)~\cite{Desroches2012} separated by chaotic windows.
To describe these MMOs, we use the following notation: \(h_1^{l_1}h_2^{l_2}\ldots\), which indicates that one period consists of \(h_1\) high-amplitude oscillations, followed by \(l_1\) low-amplitude oscillations, followed by \(h_2\) high-amplitude oscillations, and so on.

We chose an initial condition near the attractor, discarded transients and integrated this system's dynamics for \(n=10^5\) time units, sampling each time unit, with several adaptive Runge--Kutta methods, namely Kutta's 3\textsuperscript{rd}-order method (using the midpoint method for error estimation)\footnote{In the GSL's source code this method is referred to as \emph{Euler--Cauchy}, but this name is predominantly used for the classical (1\textsuperscript{st}-order) Euler method.
Moreover, the 3\textsuperscript{rd}-order method seems to go back to Kutta, who derived it in Ref.~\onlinecite{Kutta1901}.}\nocite{Kutta1901}, the classical Runge--Kutta method (with step doubling used for error estimation), Fehlberg's 4\textsuperscript{th}-order method, the Cash--Karp method, and Dormand's and Prince's 8\textsuperscript{th}-order method -- all as implemented in the GNU Scientific Library~\cite{Galassi2009}.
We applied our test to the temporal evolution of the first oscillator's first dynamical variable (\(x_1\)~in Ref.~\onlinecite{Karnatak2014}), whose maximum absolute value was ca.~\(0.9\).
Due to the latter, the highest expected absolute integration error roughly corresponds to the relative integration accuracy~\(\rho\).

For a first analysis, we chose a coupling strength \(\kappa\) of \(0.17\), for which the system exhibits a \(1^2\) MMO, i.e., a periodic dynamics (see inset of Fig.~\ref{fig:bifurcation}).
We considered the test to be successful if the respective time series complied with some period length smaller than 500 --~with the period length of the system's dynamics being roughly 287.
In Fig.~\ref{fig:integratortest}, we show the minimum error allowance \(\sigma_\text{min}\) needed to be made for the test to be successful depending on the integration accuracy~\(\rho\).
For Kutta's 3\textsuperscript{rd}-order method, \(\sigma_\text{min}\) is mostly one order of magnitude smaller than~\(\rho\) for \(10^{-7} \lessapprox \rho \lessapprox 10^{-2}\) and only larger than \(\rho\) in one case.
For all other integration methods, \(\sigma_\text{min}\) is never larger than~\(\rho\) and deviates from~\(0\) for at most a few values of \(\rho\), which exhibit no discernible pattern (except for the Fehlberg method and \(\rho>10^{-2}\)).
While a detailed investigation of this phenomenon and why it does not affect Kutta's 3\textsuperscript{rd}-order method is beyond the scope of this study, we hypothesize that it can be explained as follows:
The cases with \(\sigma_\text{min}=0\) are due to the fact that the integration error is not stochastic but systematic in nature and therefore likely to affect adjacent values of the foldation in a comparable way, thus not affecting their order.
We hypothesize that this effect is diminished if the step size is frequently adapted, which leads to the seemingly random deviations of \(\sigma_\text{min}\) from~0.
From the above results and our results from the previous subsection, we conclude that the maximum expected absolute integration error is a good choice for~\(\sigma\).

\begin{figure}
	\includegraphics{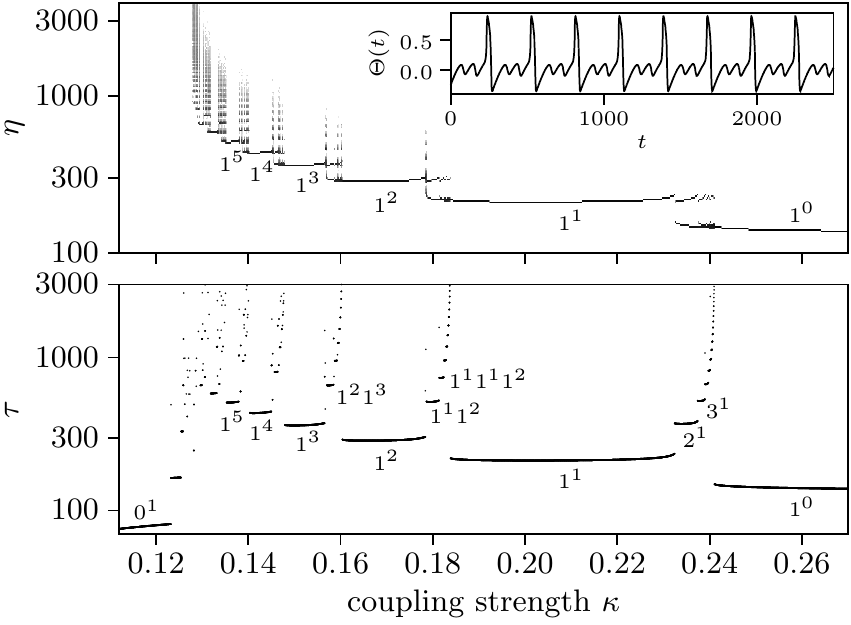}
	\caption{Top: Observed lengths of intervals \(\eta\) between subsequent high-amplitude oscillations (as identified by threshold crossings of a piecewise linear interpolation of the observed time series) depending on the coupling strength~\(\kappa\) in a system consisting of two coupled FitzHugh--Nagumo oscillators (see text) integrated with Fehlberg's 4\textsuperscript{th}-order method and a relative error \(\rho=10^{-5}\) for \(n=10^5\) time units (data already shown in Ref.~\onlinecite{Karnatak2014}).
	Bottom: Lowest period length with which the respective time series comply for an error allowance of \(\sigma=10^{-5}\).
	The occasional single dots (e.g., between the \(1^1\) and the \(1^1 1^2\) regime) occur when a sampling point falls within one of the tiny periodic windows within the small chaotic windows that separate the larger regimes.
	Labels: Mixed-mode oscillations corresponding to selected regimes.
	Inset: Excerpt of the time series for \(\kappa=0.17\) (\(1^2\) MMO).
	}
	\label{fig:bifurcation}
\end{figure}

We conclude this section with two applications of our test to identify possible periodicities in time series generated by simulated dynamical systems.
First, we investigate the coupling regimes of the aforementioned system of two diffusively coupled two-dimensional FitzHugh--Nagumo oscillators.
To this purpose, we employ the shortest period length~\(\tau\) as found by our test and the temporal distances~\(\eta\) between subsequent high-amplitude oscillations --~a marker-event-based observable.
In Fig.~\ref{fig:bifurcation}, we show the dependence of \(\tau\) and~\(\eta\) on the coupling strength~\(\kappa\).
While both allow to separate the regimes of the primary MMOs (\(1^0\), \(1^1\),~\ldots), only~\(\tau\) clearly discriminates between the regimes of secondary MMOs (\(2^1\), \(3^1\), \ldots; \(1^1 1^3\), \(1^1 1^1 1^2\), \ldots; \(1^2 1^3\), \(1^2 1^2 1^3\),~\ldots).
This demonstrates that our method may provide complementary information when analyzing dynamical regimes --~in addition to telling chaotic dynamics from periodic ones.

\begin{figure*}
	\includegraphics{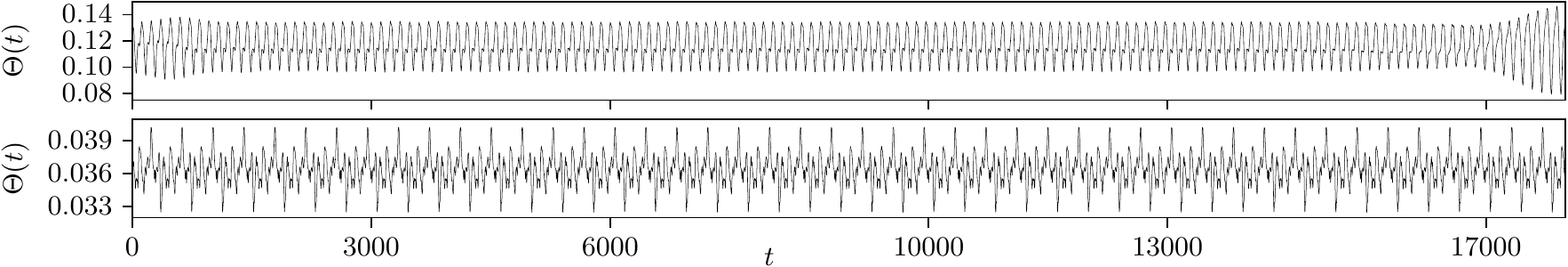}
	\caption{Two exemplary time series of the mean of the first dynamical variable of 10000 FitzHugh--Nagumo oscillators that were diffusively coupled on a small-world topology.}
	\label{fig:application}
\end{figure*}

Finally, we apply our method to time series generated by a small-world network of 10000 diffusively coupled FitzHugh--Nagumo oscillators~\cite{Ansmann2013} (the system will be discussed in detail elsewhere~\cite{Ansmann2015}).
We evolved these systems with  Fehlberg's 4\textsuperscript{th}-order method with a relative error \(\rho=10^{-5}\) and employed the average of the first dynamical variable (\(\bar{x}\) in Refs. \onlinecite{Ansmann2013} and~\onlinecite{Ansmann2015}) as an observable.
As the first dynamical variables roughly range between \(-0.4\) and \(0.9\), we expect the maximal absolute integration error to be roughly~\(\rho\).
This system is of interest here because it is capable of exhibiting long, nearly periodic episodes, which eventually turn out to be a transient behavior.

In the top of Fig.~\ref{fig:application}, we show a time series containing such an episode, whose aperiodicity becomes evident through its finiteness.
For \(1500 \lessapprox t \lessapprox 15000\) this time series appears periodic with a period length of roughly 650.
Had our observation ended at \(t=15000\), visual inspection might thus have led us to the false conclusion that the dynamics might have become a stable, periodic one.
Exemplarily applying our test to this time series for the \(t\in\kle{3000,13000}\), i.e., to the center of the periodic episode, we find that this excerpt does indeed not comply with any period length \(\tau \in \kle{2, 10000}\) for any error allowance~\(\sigma\) smaller than \(10^{-3.2}\).
In the bottom of Fig.~\ref{fig:application}, we show an excerpt of another episode, which continued for at least another \(1.5\cdot10^{6}\) time units and which we thus hypothesize to be periodic.
Applying our test to this time series for \(t \in \kle{3000,13000}\), we find that this excerpt does comply with period lengths around \(388\) for error allowances \(10^{-6} \lessapprox \sigma \lessapprox 10^{-2.2}\).
We obtain the same result when applying the test to the time series for the following \(1.5\cdot10^{6}\) time units, which affirms our hypothesis that this episode is actually periodic.
The marker-event test finds both (de-meaned) time series to be periodic and thus in particular fails to detect the aperiodicity of the first time series.
In all cases, we obtain comparable results for other, similar intervals.

\section{Conclusions}\label{conclusions}
We proposed a method to test whether, for a given time series, there is a periodic function that interpolates it and whose local extrema are captured by the time series.
Due to the conservativeness of the criterion, our method is highly specific and capable of detecting even small deviations from periodicity.
Moreover, our approach yields an interval of possible period lengths that is usually narrow in comparison to the sampling rate and allows for a precise reconstruction of one period of the observable of the time series.
We found that, in typical situations, our method outperforms an alternative, marker-event-based test in terms of specificity, precision of the detected period length, and robustness --~even though this test was tailored to the investigated time series.
By applying it to two typical problems, we also demonstrated our method's usefulness for the analysis of simulated time-continuous dynamical systems.

The first parameter that needs to be chosen for our method is the maximal accepted period length.
Choosing it too high may impair the specificity of the method and its runtime if the time series is not periodic with a short period length.
The second parameter is the error allowance, which can be straightforwardly chosen in the case of a simulated system via the error of the integration method.
The individual features of the time series do not affect the choice of either of these parameters --~while they have to be taken into account for many other methods, e.g., when choosing the marker-events for a marker-event-based approach.

For high-dimensional systems, the main computational challenge (in terms of both, runtime and stability) required to apply our method is evolving the system's dynamics to generate a time series.
Thus, if one has already performed the latter, our method can be applied with little effort---in contrast to commonly used techniques such as the maximum Lyapunov exponent or numerically finding an approximation of a periodic orbit.
The latter is often used in continuation methods~\cite{Krauskopf2007}, whom our method may assist by providing accurate starting values, namely a precise estimate of the period length as well as an approximation of the orbit by folding all dynamical variables.

It is essential for our approach that the period length does not vary considerably and that there are no phase jumps or drifts --~a requirement that is almost only fulfilled by deterministic and stationary systems.
Moreover, our test requires errors to be small and bounded with respect to features of interest.
These requirements are rarely met by real systems or experimental observations, respectively, and thus we cannot expect that our test in its present form will find application for experimental data.
Nonetheless, parts of our approach, in particular our findings on the arithmetics of folding (Th.~\ref{T:sorting}) and the nested-interval approach, may enhance existing or inspire new epoch-folding techniques.
Moreover, a variation of our approach may be applicable to the analysis of those systems that can be approximated as deterministic and stationary for a sufficiently long time with respect to the period length.
One example for such systems are pulsars~\cite{Manchester2005}, where the additional knowledge of the observable's temporal derivative may be employed to make our approach applicable to period lengths that are smaller than the sampling time~\cite{Brazier1994, Freire2001}.

\begin{acknowledgments}
	I thank U.~Feudel, P.~Freiri, R.~Karnatak, A.~Krieger, K.~Lehnertz, and J.~Schwabedal for interesting discussions and S.~Bialonski, K.~Lehnertz, S.~Porz, and A.~Saha for critical comments on earlier versions of the manuscript.
	This work was supported by the Volkswagen Foundation (Grant~Nos. 85392 and 88463)
\end{acknowledgments}

\appendix

\section{Proof of Theorem~\ref{T:main}} \label{proofT1}
\setcounter{theorem}{0}
\begin{theorem}
	For a given~\(\tau\), let \(\tilde{\mathbb{D}}_\tau \defi \klg{\alpha \in \kle{0,n-1} \mid \exists i \in \urbild: i \bmod \tau = \alpha \bmod \tau}\).
	For \(\alpha \in \tilde{\mathbb{D}}_\tau\), let \(l_\alpha \in \urbild\) be defined such that \(l_\alpha \bmod \tau = \alpha \bmod \tau\).
	Finally define \(\tilde{\Theta}_\tau: \tilde{\mathbb{D}}_\tau \rightarrow \mathbb{R}\) via \(\tilde{\Theta}_\tau \kl{\alpha} \defi \Theta\kl{l_\alpha}\).
	Then \(\Theta\) complies with a period length~\(\tau\), if and only if \(\tilde{\Theta}_\tau\) is well-defined and \(E(\tilde{\Theta}_\tau) - C(\tilde{\Theta}_\tau)= E\kl{\Theta}\).
\end{theorem}

\begin{proof}
We first note that if some function~\(\Gamma_2\) extends some other function~\(\Gamma_1\), we have \(E\kl{\Gamma_1} - C\kl{\Gamma_1} \leq E\kl{\Gamma_2} - C\kl{\Gamma_2}\).

If \(\Theta\) complies with a period length~\(\tau\), there exists a~\(\Phi\) that extends~\(\Theta\) to \(\kle{0,n-1}\) and with \(\Phi\kl{t} = \Phi\kl{t \bmod \tau}\). We therefore have:
\begin{align*}
\tilde{\Theta}_\tau\kl{\alpha}
&= \Theta\kl{l_\alpha}
= \Phi\kl{l_\alpha}
= \Phi\kl{l_\alpha \bmod \tau}\\
&= \Phi\kl{\alpha \bmod \tau}
= \Phi\kl{\alpha}
\quad \forall \alpha \in \tilde{\mathbb{D}}_\tau,
\end{align*}
and thus \(\Phi\)~also is an extension of \(\tilde{\Theta}_\tau\).
Because of this, of \(\tilde{\Theta}_\tau\)~extending~\(\Theta\), and of \(C\kl{\Theta}=0\), we have
\[E\kl{\Theta} = E\kl{\Theta} - C\kl{\Theta} \leq E(\tilde{\Theta}_\tau) - C(\tilde{\Theta}_\tau) \leq E\kl{\Phi} - C\kl{\Phi}. \]
Therefore, \(E\kl{\Phi} - C\kl{\Phi} = E\kl{\Theta}\) yields \(E(\tilde{\Theta}_\tau) - C(\tilde{\Theta}_\tau)= E\kl{\Theta}\).
For \(\tilde{\Theta}_\tau\) to be ill-defined, \(l_\alpha\) must be ill-defined and thus, there need to be \(i,j \in \urbild\) with \(i \bmod \tau = j \bmod \tau\) and \(\Theta\kl{i} \neq \Theta\kl{j}\).
From this, it directly follows that no extension of~\(\Theta\) can be \(\tau\)-periodic.

To show the other direction of the equivalence, we construct an extension~\(\tilde{\Phi}\) of~\(\tilde{\Theta}_\tau\) to \(\kle{0,n-1}\) by piecewise linear interpolation, for which we thus have \(E\kl{\Phi} - C\kl{\Phi} =  E(\tilde{\Theta}_\tau) - C(\tilde{\Theta}_\tau)\).
If \(E(\tilde{\Theta}_\tau) - C(\tilde{\Theta}_\tau)= E\kl{\Theta}\), we thus also have  \(E\kl{\Phi} - C\kl{\Phi} = E\kl{\Theta}\).
Furthermore,
\begin{align*}
\tilde{\Phi}\kl{\alpha}
&= \tilde{\Theta}_\tau\kl{\alpha}
= \Theta \kl{l_\alpha}
= \Theta \kl{l_{\alpha\bmod\tau}}\\
&= \tilde{\Theta}_\tau\kl{\alpha \bmod \tau}
= \tilde{\Phi}\kl{\alpha \bmod \tau}
\quad\forall \alpha \in \tilde{\mathbb{D}}_\tau,
\end{align*}
and due to piecewise linear interpolation, also \(\tilde{\Phi}\kl{t} = \tilde{\Phi}\kl{t \bmod \tau} ~\forall t \in \kle{0,n-1}\).
Thus, the requirements of Def.~\ref{D:periodic} are fulfilled.
\end{proof}

\section{Proof of Theorem~\ref{T:sorting}}\label{proofT2}
\setcounter{theorem}{1}
\begin{theorem}
	Let \(1<\tau<n\) and let \(\frac{p}{q}\) be the largest  and~\(\frac{r}{s}\) be the smallest reduced fraction from \Fn such that \(\frac{p}{q} < \tau \leq \frac{r}{s}\).
	Define \(I_{n,\tau}\kl{i} \defi ir \bmod \kl{p+r}\) for \(i\in \klg{0,\ldots,p+r}\).
	Then \(I_{n,\tau}\) sorts the first~\(p+r\) integers modulo~\(\tau\).
	Moreover, \(I_{n,\tau}\) increases strictly monotonically on \urbild if \(\tau<\frac{r}{s}\).
\end{theorem}

\begin{proof}
Since \Fn is the element-wise inverse of the \(n\)-th Farey sequence~\(\mathbb{F}_n\), we directly get from the theory of Farey sequences
\cite{*[{Theorem~28 in }][] Hardy1979}:
\begin{equation}\label{eq:ps1rq}
	ps + 1 = rq
\end{equation}
Thus \(p\) and \(r\) are coprime and consequently so are \(r\) and \(p+r\).
From this we get that for all \(i \neq j;~ i,j < p+r\):
\[I_{n,\tau}\kl{i} = ir \bmod \kl{p+r} \neq jr \bmod \kl{p+r} = I_{n,\tau}\kl{j},\]
i.e., that \(I_{n,\tau}\) is a bijection on \(\klg{0,\ldots,p+r-1}\)
\cite{*[{Proposition 2.1.13 in }][] Stein2008}.

From the definition of \(p\), \(q\), \(r\), and \(s\), we get:
\begin{align}
	-p \bmod \tau = q\tau - p > 0 \label{eq:qtaump} \\
	r \bmod \tau = r - s \tau \geq 0 \label{eq:rmstau} .
\end{align}
Now, let \(k \in \klg{0, \ldots, p+r-1}\) and
\begin{equation}\label{eq:dkklr}
	d_k \defi \frac{kr - kr \bmod \kl{p+r}}{p+r} = \floor{\frac{kr}{p+r}} \leq r-1.
\end{equation}
As \(d_k\) and \(k-d_k\) are monotonically increasing with~\(k\), we can make the following estimate for \(k-d_k\) by inserting the lowest and highest value for~\(k\):
\begin{equation}\label{eq:kmdkleqp}
	0 \leq k-d_k \leq \kl{p+r-1} - \kl{r-1} = p.
\end{equation}
Equations~\ref{eq:qtaump} to~\ref{eq:kmdkleqp} allow us to show the following inequality:
\begin{equation} \label{eq:kltau}
\begin{split}
	\kl{k-d_k} \kl{r \bmod \tau} + d_k\kl{-p \bmod \tau}\\
	~<~ p \kl{r-s\tau} + r \kl{q\tau -p}
	~=~ \kl{rq-ps} \tau
	~\refrel{eq:ps1rq}{=}~ \tau.
\end{split}
\end{equation}
This also gives us that both, \(\kl{k-d_k} \kl{r \bmod \tau}\) and \(d_k\kl{-p \bmod \tau}\), are smaller than~\(\tau\).
Using
\begin{align*}
m \in \mathbb{N};~
\beta \in \mathbb{R};~
m \kl{\beta \bmod \tau} < \tau\\
\Rightarrow
m\beta \bmod \tau = m \kl{\beta \bmod \tau},
\end{align*}
we can thus write:
\begin{align}
	\label{eq:kltau1}
	\kl{k - d_k} r \bmod \tau
	= \kl{k - d_k} \kl{r \bmod \tau}\\
	\label{eq:kltau2}
	d_k \cdot \kl{-p} \bmod \tau
	= d_k\kl{-p \bmod \tau}.
\end{align}

Finally, we can write:
\newcommand{\gleich}{\mathrel{\makebox[\widthof{$\dblrefrel{eq:kltau1}{eq:kltau2}{=}$}]{$=$}}}
\begin{align*}
	&\hphantom{\gleich I\!} I_{n,\tau}\kl{k} \bmod \tau\\
	&\gleich \kle{kr \bmod \kl{p+r}} \bmod \tau \\
	&\mathrel{\makebox[\widthof{$\dblrefrel{eq:kltau1}{eq:kltau2}{=}$}]{$\refrel{eq:dkklr}{=}$}} \kle{kr - d_k \kl{p+r}} \bmod \tau \\
	&\gleich \kle{\kl{k - d_k} r + d_k \cdot \kl{-p}} \bmod \tau\\
	&\gleich \kle{ \kl{k - d_k} r \bmod \tau + d_k \cdot \kl{-p} \bmod \tau } \bmod \tau \\
	&\dblrefrel{eq:kltau1}{eq:kltau2}{=} \kle{\kl{k - d_k} \kl{r \bmod \tau} + d_k\kl{-p \bmod \tau} } \bmod \tau \\
	&\mathrel{\makebox[\widthof{$\dblrefrel{eq:kltau1}{eq:kltau2}{=}$}]{$\refrel{eq:kltau}{=}$}} \kl{k - d_k} \kl{r \bmod \tau} + d_k\kl{-p \bmod \tau}.
\end{align*}
Since \(k-d_k\) and \(d_k\) are both monotonically increasing with~\(k\) and at least one of them increases if \(k\)~is increased by~\(1\), \(I_{n,\tau}\kl{k} \bmod \tau\) is strictly monotonically increasing, unless \(\tau=\frac{r}{s}\) and thus \(r \bmod \tau = 0\), in which case it is only weakly monotonically increasing.
\end{proof}

\section{Asymptotic runtime behavior}\label{runtime}

We first estimate the behavior of the average runtime, if \(2 < \tau \ll n\) is the shortest period length that \(\Theta\)~complies with.
To this purpose we employ the following facts, approximations and assumptions:
\renewcommand{\theenumi}{\Alph{enumi}}
\renewcommand{\labelenumi}{(\theenumi)}
\begin{enumerate}
	\item \label{ass:basic}The runtime of the checks performed in steps \ref{finaltest} and~\ref{quicktest} of the algorithm is approximately \(\mathbf{c} n\) or \(\mathbf{c} \breve{n}\), respectively, with some constant~\(\mathbf{c}\).
	This is based on the assumption that run-time reductions due to aborting the counting of extrema early because their number already suffices to reject periodicity can be accounted for by a constant factor (which is already incorporated in~\(\mathbf{c}\)).
	We approximate everything except these checks to have a runtime of~\(0\).
	\item \label{ass:FPR} We assume that, if \(\Theta\) does not comply with any period length in \(\kl{\tfrac{p}{q}, \tfrac{r}{s}}\), the check in step~\ref{quicktest} is positive with a probability \(\phi\) and that this is independent of other results.
	We further assume that \(\phi<\tfrac{1}{2}\), which we could confirm empirically for exemplary time series without often repeating values, i.e., for which there were no pairwise different \(t_1, \ldots, t_v\) such that \(\Theta\kl{t_1}=\ldots=\Theta\kl{t_v}\) and \(\tfrac{v}{n} \gg 0\).
	\item \label{ass:leftright} We assume that, at some level of the binary search, \(\tau\) is equally likely to be in the left or right branch of the Stern--Brocot tree and this is independent of other results.
	\item \label{ass:statistical_length} On any level~\(j\) of the Stern--Brocot tree, we approximate that, for each interval, \(\breve{n}\)~corresponds to the average value of~\(\breve{n}\) over all intervals on this level.
	We denote this average value by~\(\bar{\breve{n}}_j\).
	This approximation is based on the assumption that it is essentially at random which intervals are investigated.
	\item Let \(h\) be the level at which the smallest interval containing~\(\tau\) involved in the search resides.
	Then we assume that no higher level than~\(h\) is involved in the search.
	\item One can neglect the additional checks that are made to ensure that a maximal interval is found such that \(\Theta\) complies with all period lengths in that interval.
\end{enumerate}

Going by these assumptions, we now first estimate the number~\(a_j\) of checks performed on the \(j\)-th level (of the employed branch of the Stern--Brocot tree):
At each level between \(1\) and~\(h\), we perform one check for the interval containing~\(\tau\).
At each level between \(2\) and~\(h\) an additional check is performed with a probability of \(\tfrac{1}{2}\) (if \(\tau\) is in the right branch; see assumption~\ref{ass:leftright}).
Each of these checks has a probability of~\(\phi\) to cause two additional checks an the next level, each of which in turn has a probability of~\(\phi\) to cause two additional checks on the next level and so forth (assumption~\ref{ass:FPR}).
We therefore obtain on average:
\begin{align*}
a_j
&= 1 + \begin{cases} 0 &\text{if } j = 1 \\
 \tfrac{1}{2}  & \text{else}\end{cases} + \sum_{i=1}^{j-2} \kl{2\phi}^i
\leq 1 + \frac{1}{2} + \sum_{i=1}^{\infty} \kl{2\phi}^i\\
&= \frac{1}{2} + \sum_{i=0}^{\infty} \kl{2\phi}^i
= \frac{1}{2} + \frac{1}{1-2\phi} \ifed \hat{a}.
\end{align*}

Now, let the numerators on some level~\(j\) of the Stern--Brocot tree be \(z_1, \ldots, z_w\).
Then the numerators on the next level are \(z_1, z_1+z_2, z_2, \ldots, z_{w-1}+z_w, z_w\) and thus (by approximation~\ref{ass:statistical_length}):
\begin{align*}
\bar{\breve{n}}_j
&= \tfrac{1}{w-1} \kl{z_1 + 2 \sum_{i=2}^{w-1} z_i + z_w}\\
&= \tfrac{2}{3} \tfrac{1}{2\kl{w-1}} \kl{3 z_1 + 6 \sum_{i=2}^{w-1} z_i + 3 z_w}
= \tfrac{2}{3} \bar{\breve{n}}_{j+1}.
\end{align*}
Using this, we obtain for the runtime~\(\mathbf{g}_j\) of each check performed on the \(j\)-th level (by approximation~\ref{ass:basic}):
\[
\mathbf{g}_j = \mathbf{c} \bar{\breve{n}}_{h-1}\kl{\frac{2}{3}}^{h-j-1} ~\forall j \in \klg{1, \ldots, h-1}
~\text{and}~
\mathbf{g}_h=\mathbf{c}n.
\]
As \(\bar{\breve{n}}_{h-1} <n\) , we can thus estimate \(\mathbf{g}_j \leq \mathbf{c} n \kl{\tfrac{2}{3}}^{h-j-1}\) for all~\(j\).
Finally, we obtain for the total runtime~\(\mathbf{r}\):
\begin{align*}
\mathbf{r}
&= \sum_{j=1}^h a_j \mathbf{g}_j
\leq \sum_{j=1}^h \hat{a} \mathbf{c} n \kl{\frac{2}{3}}^{h-j-1}
= \hat{a} \mathbf{c} n \sum_{i=-1}^{h-2} \kl{\frac{2}{3}}^i\\
&< \hat{a} \mathbf{c} n \frac{3}{2} \sum_{i=0}^\infty \kl{\frac{2}{3}}^i
= \hat{a} \mathbf{c} n \frac{3}{2} \frac{1}{1-\tfrac{2}{3}}
= \frac{9}{2} \kl{\frac{1}{2} + \frac{1}{1-2\phi}} \mathbf{c} n\\
&= \mathcal{O}\kl{n}
\qquad\text{as}\qquad
n \rightarrow \infty
\end{align*}

The above does not apply, if \(\Theta\) does not comply with any period length smaller than \(n-1\).
This is because assumption~\ref{ass:leftright} and approximation~\ref{ass:statistical_length} do not hold anymore as \(\tau\) is always in the right branch and thus it cannot be considered random which interval is investigated.
In this case, even if \(\phi=0\), we have to check the intervals \(\leftopen{\tfrac{2}{1},\tfrac{3}{1}}, \leftopen{\tfrac{3}{1},\tfrac{4}{1}}, \ldots, \leftopen{\tfrac{\ceil{\tau_\text{max}}-1}{1},\tfrac{\ceil{\tau_\text{max}}}{1}}\) and thus the total runtime is:
\begin{align*}
 \mathbf{r}
&= \sum_{i=2}^{\ceil{\tau_\text{max}}-1} \kl{2i+1} \mathbf{c}
= \kl{\ceil{\tau_\text{max}}^2 - 4} \mathbf{c}\\
&= \mathcal{O}\kl{\tau_\text{max}^2}
\qquad\text{as}\qquad
\tau_\text{max} \rightarrow \infty
\end{align*}
In particular, we have \(\mathbf{r}=\mathcal{O}(n^2)\) as \(n \rightarrow \infty\), if \(\tau_\text{max} \propto n\).
A similar approximation can be made, if the shortest period length that \(\Theta\)~complies with is close to \(n\).


\begin{thebibliography}{43}%
\makeatletter
\providecommand \@ifxundefined [1]{%
 \@ifx{#1\undefined}
}%
\providecommand \@ifnum [1]{%
 \ifnum #1\expandafter \@firstoftwo
 \else \expandafter \@secondoftwo
 \fi
}%
\providecommand \@ifx [1]{%
 \ifx #1\expandafter \@firstoftwo
 \else \expandafter \@secondoftwo
 \fi
}%
\providecommand \natexlab [1]{#1}%
\providecommand \enquote  [1]{``#1''}%
\providecommand \bibnamefont  [1]{#1}%
\providecommand \bibfnamefont [1]{#1}%
\providecommand \citenamefont [1]{#1}%
\providecommand \href@noop [0]{\@secondoftwo}%
\providecommand \href [0]{\begingroup \@sanitize@url \@href}%
\providecommand \@href[1]{\@@startlink{#1}\@@href}%
\providecommand \@@href[1]{\endgroup#1\@@endlink}%
\providecommand \@sanitize@url [0]{\catcode `\\12\catcode `\$12\catcode
  `\&12\catcode `\#12\catcode `\^12\catcode `\_12\catcode `\%12\relax}%
\providecommand \@@startlink[1]{}%
\providecommand \@@endlink[0]{}%
\providecommand \url  [0]{\begingroup\@sanitize@url \@url }%
\providecommand \@url [1]{\endgroup\@href {#1}{\urlprefix }}%
\providecommand \urlprefix  [0]{URL }%
\providecommand \Eprint [0]{\href }%
\providecommand \doibase [0]{http://dx.doi.org/}%
\providecommand \selectlanguage [0]{\@gobble}%
\providecommand \bibinfo  [0]{\@secondoftwo}%
\providecommand \bibfield  [0]{\@secondoftwo}%
\providecommand \translation [1]{[#1]}%
\providecommand \BibitemOpen [0]{}%
\providecommand \bibitemStop [0]{}%
\providecommand \bibitemNoStop [0]{.\EOS\space}%
\providecommand \EOS [0]{\spacefactor3000\relax}%
\providecommand \BibitemShut  [1]{\csname bibitem#1\endcsname}%
\let\auto@bib@innerbib\@empty
\bibitem [{\citenamefont {Guckenheimer}\ and\ \citenamefont
  {Holmes}(1983)}]{Guckenheimer1983}%
  \BibitemOpen
  \bibfield  {author} {\bibinfo {author} {\bibfnamefont {J.}~\bibnamefont
  {Guckenheimer}}\ and\ \bibinfo {author} {\bibfnamefont {P.}~\bibnamefont
  {Holmes}},\ }\href {\doibase 10.1007/978-1-4612-1140-2} {\emph {\bibinfo
  {title} {Nonlinear Oscillations, Dynamical Systems, and Bifurcations of
  Vector Fields}}}\ (\bibinfo  {publisher} {Springer},\ \bibinfo {address} {New
  York},\ \bibinfo {year} {1983})\BibitemShut {NoStop}%
\bibitem [{\citenamefont {Haykin}(1983)}]{Haykin1983}%
  \BibitemOpen
  \bibinfo {editor} {\bibfnamefont {S.}~\bibnamefont {Haykin}},\ ed.,\ \href
  {\doibase 10.1007/3-540-12386-5} {\emph {\bibinfo {title} {Nonlinear methods
  of spectral analysis}}}\ (\bibinfo  {publisher} {Springer},\ \bibinfo
  {address} {Berlin, Heidelberg},\ \bibinfo {year} {1983})\BibitemShut
  {NoStop}%
\bibitem [{\citenamefont {Hale}\ and\ \citenamefont
  {Ko\c{c}ak}(1991)}]{Hale1991}%
  \BibitemOpen
  \bibfield  {author} {\bibinfo {author} {\bibfnamefont {J.~K.}\ \bibnamefont
  {Hale}}\ and\ \bibinfo {author} {\bibfnamefont {H.}~\bibnamefont
  {Ko\c{c}ak}},\ }\href {\doibase 10.1007/978-1-4612-4426-4} {\emph {\bibinfo
  {title} {Dynamics and bifurcations}}}\ (\bibinfo  {publisher}
  {Springer-Verlag},\ \bibinfo {address} {New York},\ \bibinfo {year}
  {1991})\BibitemShut {NoStop}%
\bibitem [{\citenamefont {Strogatz}(1994)}]{Strogatz1994}%
  \BibitemOpen
  \bibfield  {author} {\bibinfo {author} {\bibfnamefont {S.~H.}\ \bibnamefont
  {Strogatz}},\ }\href@noop {} {\emph {\bibinfo {title} {Nonlinear dynamics and
  chaos: with applications to physics, biology, chemistry, and engineering}}}\
  (\bibinfo  {publisher} {Addison--Wesley},\ \bibinfo {address} {Reading},\
  \bibinfo {year} {1994})\BibitemShut {NoStop}%
\bibitem [{\citenamefont {Ott}(2002)}]{Ott2002}%
  \BibitemOpen
  \bibfield  {author} {\bibinfo {author} {\bibfnamefont {E.}~\bibnamefont
  {Ott}},\ }\href {\doibase 10.1017/CBO9780511803260} {\emph {\bibinfo {title}
  {Chaos in Dynamical Systems}}},\ \bibinfo {edition} {2nd}\ ed.\ (\bibinfo
  {publisher} {Cambridge University Press},\ \bibinfo {address} {Cambridge},\
  \bibinfo {year} {2002})\BibitemShut {NoStop}%
\bibitem [{\citenamefont {Kantz}\ and\ \citenamefont
  {Schreiber}(2003)}]{Kantz2003}%
  \BibitemOpen
  \bibfield  {author} {\bibinfo {author} {\bibfnamefont {H.}~\bibnamefont
  {Kantz}}\ and\ \bibinfo {author} {\bibfnamefont {T.}~\bibnamefont
  {Schreiber}},\ }\href {\doibase 10.1017/CBO9780511755798} {\emph {\bibinfo
  {title} {Nonlinear Time Series Analysis}}},\ \bibinfo {edition} {2nd}\ ed.\
  (\bibinfo  {publisher} {Cambridge University Press},\ \bibinfo {address}
  {Cambridge, UK},\ \bibinfo {year} {2003})\BibitemShut {NoStop}%
\bibitem [{\citenamefont {Gottwald}\ and\ \citenamefont
  {Skokos}(2014)}]{Gottwald2014}%
  \BibitemOpen
  \bibfield  {author} {\bibinfo {author} {\bibfnamefont {G.~A.}\ \bibnamefont
  {Gottwald}}\ and\ \bibinfo {author} {\bibfnamefont {C.}~\bibnamefont
  {Skokos}},\ }\bibfield  {title} {\enquote {\bibinfo {title} {Preface to the
  focus issue: Chaos detection methods and predictability},}\ }\href {\doibase
  http://dx.doi.org/10.1063/1.4884603} {\bibfield  {journal} {\bibinfo
  {journal} {Chaos}\ }\textbf {\bibinfo {volume} {24}},\ \bibinfo {pages}
  {024201} (\bibinfo {year} {2014})}\BibitemShut {NoStop}%
\bibitem [{\citenamefont {Burki}, \citenamefont {Maeder},\ and\ \citenamefont
  {Rufener}(1978)}]{Burki1978}%
  \BibitemOpen
  \bibfield  {author} {\bibinfo {author} {\bibfnamefont {G.}~\bibnamefont
  {Burki}}, \bibinfo {author} {\bibfnamefont {A.}~\bibnamefont {Maeder}}, \
  and\ \bibinfo {author} {\bibfnamefont {F.}~\bibnamefont {Rufener}},\
  }\bibfield  {title} {\enquote {\bibinfo {title} {Variable stars of small
  amplitude {III}. {Semi}-period of variation for seven {B2} to {G0} supergiant
  stars},}\ }\href@noop {} {\bibfield  {journal} {\bibinfo  {journal} {Astron.
  Astrophys.}\ }\textbf {\bibinfo {volume} {65}},\ \bibinfo {pages} {363--367}
  (\bibinfo {year} {1978})}\BibitemShut {NoStop}%
\bibitem [{\citenamefont {Siegel}(1980)}]{Siegel1980}%
  \BibitemOpen
  \bibfield  {author} {\bibinfo {author} {\bibfnamefont {A.~F.}\ \bibnamefont
  {Siegel}},\ }\bibfield  {title} {\enquote {\bibinfo {title} {Testing for
  periodicity in a time series},}\ }\href {\doibase
  10.1080/01621459.1980.10477474} {\bibfield  {journal} {\bibinfo  {journal}
  {J. Am. Stat. Assoc.}\ }\textbf {\bibinfo {volume} {75}},\ \bibinfo {pages}
  {345--348} (\bibinfo {year} {1980})}\BibitemShut {NoStop}%
\bibitem [{\citenamefont {Scargle}(1982)}]{Scargle1982}%
  \BibitemOpen
  \bibfield  {author} {\bibinfo {author} {\bibfnamefont {J.~D.}\ \bibnamefont
  {Scargle}},\ }\bibfield  {title} {\enquote {\bibinfo {title} {Studies in
  astronomical time series analysis. {II.} {Statistical} aspects of spectral
  analysis of unevenly spaced data},}\ }\href {\doibase 10.1086/160554}
  {\bibfield  {journal} {\bibinfo  {journal} {Astrophs.~J.}\ }\textbf {\bibinfo
  {volume} {263}},\ \bibinfo {pages} {835--853} (\bibinfo {year}
  {1982})}\BibitemShut {NoStop}%
\bibitem [{\citenamefont {Vlachos}, \citenamefont {Yu},\ and\ \citenamefont
  {Castelli}(2005)}]{Vlachos2005}%
  \BibitemOpen
  \bibfield  {author} {\bibinfo {author} {\bibfnamefont {M.}~\bibnamefont
  {Vlachos}}, \bibinfo {author} {\bibfnamefont {P.}~\bibnamefont {Yu}}, \ and\
  \bibinfo {author} {\bibfnamefont {V.}~\bibnamefont {Castelli}},\ }\bibfield
  {title} {\enquote {\bibinfo {title} {On periodicity detection and structural
  periodic similarity},}\ }in\ \href {\doibase 10.1137/1.9781611972757.40}
  {\emph {\bibinfo {booktitle} {Proceedings of the 2005 SIAM International
  Conference on Data Mining}}}\ (\bibinfo  {publisher} {SIAM},\ \bibinfo
  {address} {Philadelphia},\ \bibinfo {year} {2005})\ Chap.~\bibinfo {chapter}
  {40}, pp.\ \bibinfo {pages} {449--460}\BibitemShut {NoStop}%
\bibitem [{\citenamefont {Heck}, \citenamefont {Manfroid},\ and\ \citenamefont
  {Mersch}(1985)}]{Heck1985}%
  \BibitemOpen
  \bibfield  {author} {\bibinfo {author} {\bibfnamefont {A.}~\bibnamefont
  {Heck}}, \bibinfo {author} {\bibfnamefont {J.}~\bibnamefont {Manfroid}}, \
  and\ \bibinfo {author} {\bibfnamefont {G.}~\bibnamefont {Mersch}},\
  }\bibfield  {title} {\enquote {\bibinfo {title} {On period determination
  methods},}\ }\href@noop {} {\bibfield  {journal} {\bibinfo  {journal}
  {Astron. Astrophys. Suppl. Seri.}\ }\textbf {\bibinfo {volume} {59}},\
  \bibinfo {pages} {63--72} (\bibinfo {year} {1985})}\BibitemShut {NoStop}%
\bibitem [{\citenamefont {Davies}(1990)}]{Davies1990}%
  \BibitemOpen
  \bibfield  {author} {\bibinfo {author} {\bibfnamefont {S.~R.}\ \bibnamefont
  {Davies}},\ }\bibfield  {title} {\enquote {\bibinfo {title} {An improved test
  for periodicity},}\ }\href@noop {} {\bibfield  {journal} {\bibinfo  {journal}
  {Mon. Not. R. Astr. Soc.}\ }\textbf {\bibinfo {volume} {244}},\ \bibinfo
  {pages} {93--95} (\bibinfo {year} {1990})}\BibitemShut {NoStop}%
\bibitem [{\citenamefont {Cincotta}\ \emph {et~al.}(1999)\citenamefont
  {Cincotta}, \citenamefont {Helmi}, \citenamefont {M\'endez}, \citenamefont
  {N\'u\~nez},\ and\ \citenamefont {Vucetich}}]{Cincotta1999}%
  \BibitemOpen
  \bibfield  {author} {\bibinfo {author} {\bibfnamefont {P.~M.}\ \bibnamefont
  {Cincotta}}, \bibinfo {author} {\bibfnamefont {A.}~\bibnamefont {Helmi}},
  \bibinfo {author} {\bibfnamefont {M.}~\bibnamefont {M\'endez}}, \bibinfo
  {author} {\bibfnamefont {J.~A.}\ \bibnamefont {N\'u\~nez}}, \ and\ \bibinfo
  {author} {\bibfnamefont {H.}~\bibnamefont {Vucetich}},\ }\bibfield  {title}
  {\enquote {\bibinfo {title} {Astronomical time-series analysis -- {II.} {A}
  search for periodicity using the {Shannon} entropy},}\ }\href {\doibase
  10.1046/j.1365-8711.1999.02128.x} {\bibfield  {journal} {\bibinfo  {journal}
  {Mon. Not. R. Astr. Soc.}\ }\textbf {\bibinfo {volume} {302}},\ \bibinfo
  {pages} {582--586} (\bibinfo {year} {1999})}\BibitemShut {NoStop}%
\bibitem [{\citenamefont {Larsson}(1996)}]{Larsson1996}%
  \BibitemOpen
  \bibfield  {author} {\bibinfo {author} {\bibfnamefont {S.}~\bibnamefont
  {Larsson}},\ }\bibfield  {title} {\enquote {\bibinfo {title} {Parameter
  estimation in epoch folding analysis},}\ }\href {\doibase
  a10.1051/aas:1996150} {\bibfield  {journal} {\bibinfo  {journal} {Astron.
  Astrophys. Suppl. Ser.}\ }\textbf {\bibinfo {volume} {117}},\ \bibinfo
  {pages} {197--201} (\bibinfo {year} {1996})}\BibitemShut {NoStop}%
\bibitem [{\citenamefont {Han}, \citenamefont {Dong},\ and\ \citenamefont
  {Yin}(1999)}]{Han1999}%
  \BibitemOpen
  \bibfield  {author} {\bibinfo {author} {\bibfnamefont {J.}~\bibnamefont
  {Han}}, \bibinfo {author} {\bibfnamefont {G.}~\bibnamefont {Dong}}, \ and\
  \bibinfo {author} {\bibfnamefont {Y.}~\bibnamefont {Yin}},\ }\bibfield
  {title} {\enquote {\bibinfo {title} {Efficient mining of partial periodic
  patterns in time series database},}\ }in\ \href {\doibase
  10.1109/ICDE.1999.754913} {\emph {\bibinfo {booktitle} {Proceedings of the
  15\textsuperscript{th} International Conference on Data Engineering}}}\
  (\bibinfo  {publisher} {IEEE},\ \bibinfo {address} {Los Alamitos},\ \bibinfo
  {year} {1999})\ pp.\ \bibinfo {pages} {106--115}\BibitemShut {NoStop}%
\bibitem [{\citenamefont {Ergun}, \citenamefont {Muthukrishnan},\ and\
  \citenamefont {Sahinalp}(2004)}]{Ergun2004}%
  \BibitemOpen
  \bibfield  {author} {\bibinfo {author} {\bibfnamefont {F.}~\bibnamefont
  {Ergun}}, \bibinfo {author} {\bibfnamefont {S.}~\bibnamefont
  {Muthukrishnan}}, \ and\ \bibinfo {author} {\bibfnamefont {S.~C.}\
  \bibnamefont {Sahinalp}},\ }\bibfield  {title} {\enquote {\bibinfo {title}
  {Sublinear methods for detecting periodic trends in data streams},}\ }in\
  \href {\doibase 10.1007/978-3-540-24698-5_6} {\emph {\bibinfo {booktitle}
  {LATIN 2004: Theoretical Informatics}}},\ \bibinfo {series} {Lecture Notes in
  Computer Science}, Vol.\ \bibinfo {volume} {2976},\ \bibinfo {editor} {edited
  by\ \bibinfo {editor} {\bibfnamefont {M.}~\bibnamefont {Farach-Colton}}}\
  (\bibinfo  {publisher} {Springer},\ \bibinfo {address} {Berlin, Heidelberg},\
  \bibinfo {year} {2004})\ pp.\ \bibinfo {pages} {16--28}\BibitemShut {NoStop}%
\bibitem [{\citenamefont {Elfeky}, \citenamefont {Aref},\ and\ \citenamefont
  {Elmagarmid}(2005)}]{Elfeky2005}%
  \BibitemOpen
  \bibfield  {author} {\bibinfo {author} {\bibfnamefont {M.~G.}\ \bibnamefont
  {Elfeky}}, \bibinfo {author} {\bibfnamefont {W.~G.}\ \bibnamefont {Aref}}, \
  and\ \bibinfo {author} {\bibfnamefont {A.~K.}\ \bibnamefont {Elmagarmid}},\
  }\bibfield  {title} {\enquote {\bibinfo {title} {Periodicity detection in
  time series databases},}\ }\href {\doibase 10.1109/TKDE.2005.114} {\bibfield
  {journal} {\bibinfo  {journal} {IEEE Trans. Knowl. Data Eng.}\ }\textbf
  {\bibinfo {volume} {17}},\ \bibinfo {pages} {875--887} (\bibinfo {year}
  {2005})}\BibitemShut {NoStop}%
\bibitem [{\citenamefont {Rosenblum}\ and\ \citenamefont
  {Kurths}(1995)}]{Rosenblum1995}%
  \BibitemOpen
  \bibfield  {author} {\bibinfo {author} {\bibfnamefont {M.}~\bibnamefont
  {Rosenblum}}\ and\ \bibinfo {author} {\bibfnamefont {J.}~\bibnamefont
  {Kurths}},\ }\bibfield  {title} {\enquote {\bibinfo {title} {A simple test
  for hidden periodicity in time series data},}\ }\href {\doibase
  10.1142/S0218127495000211} {\bibfield  {journal} {\bibinfo  {journal} {Int.
  J. Bifurcat. Chaos}\ }\textbf {\bibinfo {volume} {05}},\ \bibinfo {pages}
  {265--269} (\bibinfo {year} {1995})}\BibitemShut {NoStop}%
\bibitem [{\citenamefont {Takens}(1981)}]{Takens1981}%
  \BibitemOpen
  \bibfield  {author} {\bibinfo {author} {\bibfnamefont {F.}~\bibnamefont
  {Takens}},\ }\bibfield  {title} {\enquote {\bibinfo {title} {Detecting
  strange attractors in turbulence},}\ }in\ \href {\doibase 10.1007/BFb0091924}
  {\emph {\bibinfo {booktitle} {Dynamical Systems and Turbulence (Warwick
  1980)}}},\ \bibinfo {series} {Lecture Notes in Mathematics}, Vol.\ \bibinfo
  {volume} {898},\ \bibinfo {editor} {edited by\ \bibinfo {editor}
  {\bibfnamefont {D.~A.}\ \bibnamefont {Rand}}\ and\ \bibinfo {editor}
  {\bibfnamefont {L.-S.}\ \bibnamefont {Young}}}\ (\bibinfo  {publisher}
  {Springer-Verlag},\ \bibinfo {address} {Berlin},\ \bibinfo {year} {1981})\
  pp.\ \bibinfo {pages} {366--381}\BibitemShut {NoStop}%
\bibitem [{\citenamefont {Marwan}\ \emph {et~al.}(2007)\citenamefont {Marwan},
  \citenamefont {Romano}, \citenamefont {Thiel},\ and\ \citenamefont
  {Kurths}}]{Marwan2007}%
  \BibitemOpen
  \bibfield  {author} {\bibinfo {author} {\bibfnamefont {N.}~\bibnamefont
  {Marwan}}, \bibinfo {author} {\bibfnamefont {M.~C.}\ \bibnamefont {Romano}},
  \bibinfo {author} {\bibfnamefont {M.}~\bibnamefont {Thiel}}, \ and\ \bibinfo
  {author} {\bibfnamefont {J.}~\bibnamefont {Kurths}},\ }\bibfield  {title}
  {\enquote {\bibinfo {title} {Recurrence plots for the analysis of complex
  systems},}\ }\href {\doibase 10.1016/j.physrep.2006.11.001} {\bibfield
  {journal} {\bibinfo  {journal} {Phys. Rep.}\ }\textbf {\bibinfo {volume}
  {438}},\ \bibinfo {pages} {237--329} (\bibinfo {year} {2007})}\BibitemShut
  {NoStop}%
\bibitem [{\citenamefont {Brazier}(1994)}]{Brazier1994}%
  \BibitemOpen
  \bibfield  {author} {\bibinfo {author} {\bibfnamefont {K.~T.~S.}\
  \bibnamefont {Brazier}},\ }\bibfield  {title} {\enquote {\bibinfo {title}
  {Confidence intervals from the {Rayleigh} test},}\ }\href {\doibase
  10.1093/mnras/268.3.709} {\bibfield  {journal} {\bibinfo  {journal} {Mon.
  Not. R. Astr. Soc.}\ }\textbf {\bibinfo {volume} {268}},\ \bibinfo {pages}
  {709--712} (\bibinfo {year} {1994})}\BibitemShut {NoStop}%
\bibitem [{\citenamefont {Freire}, \citenamefont {Kramer},\ and\ \citenamefont
  {Lyne}(2001)}]{Freire2001}%
  \BibitemOpen
  \bibfield  {author} {\bibinfo {author} {\bibfnamefont {P.~C.}\ \bibnamefont
  {Freire}}, \bibinfo {author} {\bibfnamefont {M.}~\bibnamefont {Kramer}}, \
  and\ \bibinfo {author} {\bibfnamefont {A.~G.}\ \bibnamefont {Lyne}},\
  }\bibfield  {title} {\enquote {\bibinfo {title} {Determination of the orbital
  parameters of binary pulsars},}\ }\href {\doibase
  10.1046/j.1365-8711.2001.04200.x} {\bibfield  {journal} {\bibinfo  {journal}
  {Mon. Not. R. Astr. Soc.}\ }\textbf {\bibinfo {volume} {322}},\ \bibinfo
  {pages} {885--890} (\bibinfo {year} {2001})}\BibitemShut {NoStop}%
\bibitem [{\citenamefont {Zou}\ \emph {et~al.}(2010)\citenamefont {Zou},
  \citenamefont {Donner}, \citenamefont {Donges}, \citenamefont {Marwan},\ and\
  \citenamefont {Kurths}}]{Zou2010}%
  \BibitemOpen
  \bibfield  {author} {\bibinfo {author} {\bibfnamefont {Y.}~\bibnamefont
  {Zou}}, \bibinfo {author} {\bibfnamefont {R.~V.}\ \bibnamefont {Donner}},
  \bibinfo {author} {\bibfnamefont {J.~F.}\ \bibnamefont {Donges}}, \bibinfo
  {author} {\bibfnamefont {N.}~\bibnamefont {Marwan}}, \ and\ \bibinfo {author}
  {\bibfnamefont {J.}~\bibnamefont {Kurths}},\ }\bibfield  {title} {\enquote
  {\bibinfo {title} {Identifying complex periodic windows in continuous-time
  dynamical systems using recurrence-based methods},}\ }\href {\doibase
  10.1063/1.3523304} {\bibfield  {journal} {\bibinfo  {journal} {Chaos}\
  }\textbf {\bibinfo {volume} {20}},\ \bibinfo {pages} {043130} (\bibinfo
  {year} {2010})}\BibitemShut {NoStop}%
\bibitem [{\citenamefont {Krauskopf}, \citenamefont {Osinga},\ and\
  \citenamefont {Gal\'an-Vioque}(2007)}]{Krauskopf2007}%
  \BibitemOpen
  \bibfield  {author} {\bibinfo {author} {\bibfnamefont {B.}~\bibnamefont
  {Krauskopf}}, \bibinfo {author} {\bibfnamefont {H.~M.}\ \bibnamefont
  {Osinga}}, \ and\ \bibinfo {author} {\bibfnamefont {J.}~\bibnamefont
  {Gal\'an-Vioque}},\ }\href {\doibase 10.1007/978-1-4020-6356-5} {\emph
  {\bibinfo {title} {Numerical continuation methods for dynamical systems}}}\
  (\bibinfo  {publisher} {Springer},\ \bibinfo {address} {Dordrecht},\ \bibinfo
  {year} {2007})\BibitemShut {NoStop}%
\bibitem [{\citenamefont {Benettin}\ \emph {et~al.}(1980)\citenamefont
  {Benettin}, \citenamefont {Galgani}, \citenamefont {Giorgilli},\ and\
  \citenamefont {Strelcyn}}]{Benettin1980}%
  \BibitemOpen
  \bibfield  {author} {\bibinfo {author} {\bibfnamefont {G.}~\bibnamefont
  {Benettin}}, \bibinfo {author} {\bibfnamefont {L.}~\bibnamefont {Galgani}},
  \bibinfo {author} {\bibfnamefont {A.}~\bibnamefont {Giorgilli}}, \ and\
  \bibinfo {author} {\bibfnamefont {J.-M.}\ \bibnamefont {Strelcyn}},\
  }\bibfield  {title} {\enquote {\bibinfo {title} {Lyapunov characteristic
  exponents for smooth dynamical systems and for {Hamiltonian} systems; a
  method for computing all of them},}\ }\href {\doibase 10.1007/BF02128236}
  {\bibfield  {journal} {\bibinfo  {journal} {Meccanica}\ }\textbf {\bibinfo
  {volume} {15}},\ \bibinfo {pages} {9--30} (\bibinfo {year}
  {1980})}\BibitemShut {NoStop}%
\bibitem [{Sou()}]{Source}%
  \BibitemOpen
  \href@noop {} {}\bibinfo {howpublished}
  {\texttt{https://github.com/neurophysik/periodicitytest}}\BibitemShut
  {NoStop}%
\bibitem [{\citenamefont {Vardy}(1991)}]{Vardy1991}%
  \BibitemOpen
  \bibfield  {author} {\bibinfo {author} {\bibfnamefont {I.}~\bibnamefont
  {Vardy}},\ }\href@noop {} {\emph {\bibinfo {title} {Computational Recreations
  in Mathematica}}}\ (\bibinfo  {publisher} {Addison--Wesley},\ \bibinfo
  {address} {Rewood City},\ \bibinfo {year} {1991})\BibitemShut {NoStop}%
\bibitem [{Note1()}]{Note1}%
  \BibitemOpen
  \bibinfo {note} {The Stern--Brocot~\cite {Graham1989} tree is a tree spanning
  all reduced fractions, which can be recursively defined as follows: Let
  \(\protect \mathfrak {s}^i_1, \protect \ldots , \protect \mathfrak
  {s}^i_{2^i+1}\) denote the fractions on the \(i\)-th level of the tree. Then
  \(\protect \mathfrak {s}^0_1 \mathrel {\mathop :}=\protect \genfrac
  {}{}{}1{-1}{0}\), \(\protect \mathfrak {s}^0_2 \mathrel {\mathop :}=\protect
  \genfrac {}{}{}1{1}{0}\), \(\protect \mathfrak {s}^i_{2j} \mathrel {\mathop
  :}=\protect \mathfrak {s}^{i-1}_j\), and \(\protect \mathfrak {s}^i_{2j+1}
  \mathrel {\mathop :}=\protect \mathfrak {M}{\setbox \z@ \hbox
  {\frozen@everymath \@emptytoks \mathsurround \z@ $\nulldelimiterspace \z@
  \left (\vcenter to\@ne \big@size {}\right .$}\box \z@ }\protect \mathfrak
  {s}^{i-1}_j, \protect \mathfrak {s}^{i-1}_{j+1}{\setbox \z@ \hbox
  {\frozen@everymath \@emptytoks \mathsurround \z@ $\nulldelimiterspace \z@
  \left )\vcenter to\@ne \big@size {}\right .$}\box \z@ }\), where \(\protect
  \mathfrak {M}\) denotes the \protect \emph {mediant:} \(\protect \mathfrak
  {M}{\setbox \z@ \hbox {\frozen@everymath \@emptytoks \mathsurround \z@
  $\nulldelimiterspace \z@ \left (\vcenter to\@ne \big@size {}\right .$}\box
  \z@ }\protect \genfrac {}{}{}1{x}{y},\protect \genfrac {}{}{}1{z}{w}{\setbox
  \z@ \hbox {\frozen@everymath \@emptytoks \mathsurround \z@
  $\nulldelimiterspace \z@ \left )\vcenter to\@ne \big@size {}\right .$}\box
  \z@ } \mathrel {\mathop :}=\protect \genfrac {}{}{}1{x+z}{y+w}\). Each level
  contains fractions in ascending order, i.e, \(\protect \mathfrak {s}^i_1 <
  \protect \ldots < \protect \mathfrak {s}^i_{2^i+1}\).}\BibitemShut {Stop}%
\bibitem [{\citenamefont {Graham}, \citenamefont {Knuth},\ and\ \citenamefont
  {Patashnik}(1989)}]{Graham1989}%
  \BibitemOpen
  \bibfield  {author} {\bibinfo {author} {\bibfnamefont {R.~L.}\ \bibnamefont
  {Graham}}, \bibinfo {author} {\bibfnamefont {D.~E.}\ \bibnamefont {Knuth}}, \
  and\ \bibinfo {author} {\bibfnamefont {O.}~\bibnamefont {Patashnik}},\
  }\href@noop {} {\emph {\bibinfo {title} {Concrete Mathematics: {A} Foundation
  for Computer Science}}},\ \bibinfo {edition} {2nd}\ ed.\ (\bibinfo
  {publisher} {Addison--Wesley},\ \bibinfo {address} {Upper Saddle River},\
  \bibinfo {year} {1989})\BibitemShut {NoStop}%
\bibitem [{\citenamefont {Pikovsky}, \citenamefont {Rosenblum},\ and\
  \citenamefont {Kurths}(2001)}]{Pikovsky2001}%
  \BibitemOpen
  \bibfield  {author} {\bibinfo {author} {\bibfnamefont {A.~S.}\ \bibnamefont
  {Pikovsky}}, \bibinfo {author} {\bibfnamefont {M.~G.}\ \bibnamefont
  {Rosenblum}}, \ and\ \bibinfo {author} {\bibfnamefont {J.}~\bibnamefont
  {Kurths}},\ }\href {\doibase 10.1017/CBO9780511755743} {\emph {\bibinfo
  {title} {Synchronization: {A} universal concept in nonlinear sciences}}}\
  (\bibinfo  {publisher} {Cambridge University Press},\ \bibinfo {address}
  {Cambridge, UK},\ \bibinfo {year} {2001})\BibitemShut {NoStop}%
\bibitem [{\citenamefont {Christensen}(2005)}]{Christensen2005}%
  \BibitemOpen
  \bibfield  {author} {\bibinfo {author} {\bibfnamefont {D.}~\bibnamefont
  {Christensen}},\ }\bibfield  {title} {\enquote {\bibinfo {title} {Fast
  algorithms for the calculation of {Kendall's}~\(\tau\)},}\ }\href {\doibase
  10.1007/BF02736122} {\bibfield  {journal} {\bibinfo  {journal} {Comp. Stat.}\
  }\textbf {\bibinfo {volume} {20}},\ \bibinfo {pages} {51--62} (\bibinfo
  {year} {2005})}\BibitemShut {NoStop}%
\bibitem [{\citenamefont {FitzHugh}(1961)}]{FitzHugh1961}%
  \BibitemOpen
  \bibfield  {author} {\bibinfo {author} {\bibfnamefont {R.}~\bibnamefont
  {FitzHugh}},\ }\bibfield  {title} {\enquote {\bibinfo {title} {Impulses and
  physiological states in theoretical models of nerve membrane},}\ }\href
  {\doibase 10.1016/S0006-3495(61)86902-6} {\bibfield  {journal} {\bibinfo
  {journal} {Biophys J.}\ }\textbf {\bibinfo {volume} {1}},\ \bibinfo {pages}
  {445--466} (\bibinfo {year} {1961})}\BibitemShut {NoStop}%
\bibitem [{\citenamefont {Ansmann}\ \emph {et~al.}(2013)\citenamefont
  {Ansmann}, \citenamefont {Karnatak}, \citenamefont {Lehnertz},\ and\
  \citenamefont {Feudel}}]{Ansmann2013}%
  \BibitemOpen
  \bibfield  {author} {\bibinfo {author} {\bibfnamefont {G.}~\bibnamefont
  {Ansmann}}, \bibinfo {author} {\bibfnamefont {R.}~\bibnamefont {Karnatak}},
  \bibinfo {author} {\bibfnamefont {K.}~\bibnamefont {Lehnertz}}, \ and\
  \bibinfo {author} {\bibfnamefont {U.}~\bibnamefont {Feudel}},\ }\bibfield
  {title} {\enquote {\bibinfo {title} {Extreme events in excitable systems and
  mechanisms of their generation},}\ }\href {\doibase
  10.1103/PhysRevE.88.052911} {\bibfield  {journal} {\bibinfo  {journal} {Phys.
  Rev.~E}\ }\textbf {\bibinfo {volume} {88}},\ \bibinfo {pages} {052911}
  (\bibinfo {year} {2013})}\BibitemShut {NoStop}%
\bibitem [{\citenamefont {Karnatak}\ \emph {et~al.}(2014)\citenamefont
  {Karnatak}, \citenamefont {Ansmann}, \citenamefont {Feudel},\ and\
  \citenamefont {Lehnertz}}]{Karnatak2014}%
  \BibitemOpen
  \bibfield  {author} {\bibinfo {author} {\bibfnamefont {R.}~\bibnamefont
  {Karnatak}}, \bibinfo {author} {\bibfnamefont {G.}~\bibnamefont {Ansmann}},
  \bibinfo {author} {\bibfnamefont {U.}~\bibnamefont {Feudel}}, \ and\ \bibinfo
  {author} {\bibfnamefont {K.}~\bibnamefont {Lehnertz}},\ }\bibfield  {title}
  {\enquote {\bibinfo {title} {Route to extreme events in excitable systems},}\
  }\href {\doibase 10.1103/PhysRevE.90.022917} {\bibfield  {journal} {\bibinfo
  {journal} {Phys. Rev.~E}\ }\textbf {\bibinfo {volume} {90}},\ \bibinfo
  {pages} {022917} (\bibinfo {year} {2014})}\BibitemShut {NoStop}%
\bibitem [{\citenamefont {Desroches}\ \emph {et~al.}(2012)\citenamefont
  {Desroches}, \citenamefont {Guckenheimer}, \citenamefont {Krauskopf},
  \citenamefont {Kuehn}, \citenamefont {Osinga},\ and\ \citenamefont
  {Wechselberger}}]{Desroches2012}%
  \BibitemOpen
  \bibfield  {author} {\bibinfo {author} {\bibfnamefont {M.}~\bibnamefont
  {Desroches}}, \bibinfo {author} {\bibfnamefont {J.}~\bibnamefont
  {Guckenheimer}}, \bibinfo {author} {\bibfnamefont {B.}~\bibnamefont
  {Krauskopf}}, \bibinfo {author} {\bibfnamefont {C.}~\bibnamefont {Kuehn}},
  \bibinfo {author} {\bibfnamefont {H.~M.}\ \bibnamefont {Osinga}}, \ and\
  \bibinfo {author} {\bibfnamefont {M.}~\bibnamefont {Wechselberger}},\
  }\bibfield  {title} {\enquote {\bibinfo {title} {Mixed-mode oscillations with
  multiple time scales},}\ }\href {\doibase 10.1137/100791233} {\bibfield
  {journal} {\bibinfo  {journal} {SIAM Rev.}\ }\textbf {\bibinfo {volume}
  {54}},\ \bibinfo {pages} {211--288} (\bibinfo {year} {2012})}\BibitemShut
  {NoStop}%
\bibitem [{Note2()}]{Note2}%
  \BibitemOpen
  \bibinfo {note} {In the GSL's source code this method is referred to as
  \protect \emph {Euler--Cauchy}, but this name is predominantly used for the
  classical (1\protect \textsuperscript {st}-order) Euler method. Moreover, the
  3\protect \textsuperscript {rd}-order method seems to go back to Kutta, who
  derived it in Ref.~\protect \rev@citealpnum {Kutta1901}.}\BibitemShut {Stop}%
\bibitem [{\citenamefont {Kutta}(1901)}]{Kutta1901}%
  \BibitemOpen
  \bibfield  {author} {\bibinfo {author} {\bibfnamefont {W.}~\bibnamefont
  {Kutta}},\ }\bibfield  {title} {\enquote {\bibinfo {title} {Beitrag zur
  n\"aherungsweisen {Integration} totaler {Differentialgleichungen}},}\
  }\href@noop {} {\bibfield  {journal} {\bibinfo  {journal} {Z. Math. Phys}\
  }\textbf {\bibinfo {volume} {46}},\ \bibinfo {pages} {435--453} (\bibinfo
  {year} {1901})}\BibitemShut {NoStop}%
\bibitem [{\citenamefont {Galassi}\ \emph {et~al.}(2009)\citenamefont
  {Galassi}, \citenamefont {Davies}, \citenamefont {Theiler}, \citenamefont
  {Gough}, \citenamefont {Jungman}, \citenamefont {Alken}, \citenamefont
  {Booth},\ and\ \citenamefont {Rossi}}]{Galassi2009}%
  \BibitemOpen
  \bibfield  {author} {\bibinfo {author} {\bibfnamefont {M.}~\bibnamefont
  {Galassi}}, \bibinfo {author} {\bibfnamefont {J.}~\bibnamefont {Davies}},
  \bibinfo {author} {\bibfnamefont {J.}~\bibnamefont {Theiler}}, \bibinfo
  {author} {\bibfnamefont {B.}~\bibnamefont {Gough}}, \bibinfo {author}
  {\bibfnamefont {G.}~\bibnamefont {Jungman}}, \bibinfo {author} {\bibfnamefont
  {P.}~\bibnamefont {Alken}}, \bibinfo {author} {\bibfnamefont
  {M.}~\bibnamefont {Booth}}, \ and\ \bibinfo {author} {\bibfnamefont
  {F.}~\bibnamefont {Rossi}},\ }\href@noop {} {\emph {\bibinfo {title} {{GNU}
  Scientific Library Reference Manual}}},\ \bibinfo {edition} {3rd}\ ed.\
  (\bibinfo  {publisher} {Network Theory},\ \bibinfo {address} {Bristol},\
  \bibinfo {year} {2009})\BibitemShut {NoStop}%
\bibitem [{\citenamefont {Ansmann}, \citenamefont {Lehnertz},\ and\
  \citenamefont {Feudel}()}]{Ansmann2015}%
  \BibitemOpen
  \bibfield  {author} {\bibinfo {author} {\bibfnamefont {G.}~\bibnamefont
  {Ansmann}}, \bibinfo {author} {\bibfnamefont {K.}~\bibnamefont {Lehnertz}}, \
  and\ \bibinfo {author} {\bibfnamefont {U.}~\bibnamefont {Feudel}},\
  }\href@noop {} {\enquote {\bibinfo {title} {Self-induced pattern switching on
  complex networks of excitable units},}\ }\bibinfo {note}
  {Submitted}\BibitemShut {NoStop}%
\bibitem [{\citenamefont {Manchester}\ \emph {et~al.}(2005)\citenamefont
  {Manchester}, \citenamefont {Hobbs}, \citenamefont {Teoh},\ and\
  \citenamefont {Hobbs}}]{Manchester2005}%
  \BibitemOpen
  \bibfield  {author} {\bibinfo {author} {\bibfnamefont {R.~N.}\ \bibnamefont
  {Manchester}}, \bibinfo {author} {\bibfnamefont {G.~B.}\ \bibnamefont
  {Hobbs}}, \bibinfo {author} {\bibfnamefont {A.}~\bibnamefont {Teoh}}, \ and\
  \bibinfo {author} {\bibfnamefont {M.}~\bibnamefont {Hobbs}},\ }\bibfield
  {title} {\enquote {\bibinfo {title} {The {Australia} {Telescope} {National}
  {Facility} {Pulsar} {Catalogue}},}\ }\href {\doibase 10.1086/428488}
  {\bibfield  {journal} {\bibinfo  {journal} {Astron.~J.}\ }\textbf {\bibinfo
  {volume} {129}},\ \bibinfo {pages} {1993--2006} (\bibinfo {year}
  {2005})}\BibitemShut {NoStop}%
\bibitem [{\citenamefont {Hardy}\ and\ \citenamefont
  {Wright}(1979)}]{Hardy1979}%
  \BibitemOpen
  \bibfield  {author} {\bibinfo {author} {\bibfnamefont {G.}~\bibnamefont
  {Hardy}}\ and\ \bibinfo {author} {\bibfnamefont {E.}~\bibnamefont {Wright}},\
  }\href@noop {} {\emph {\bibinfo {title} {An Introduction to the Theory of
  Numbers}}},\ \bibinfo {edition} {5th}\ ed.\ (\bibinfo  {publisher} {Oxford
  University Press},\ \bibinfo {address} {Oxford},\ \bibinfo {year}
  {1979})\BibitemShut {NoStop}%
\bibitem [{\citenamefont {Stein}(2008)}]{Stein2008}%
  \BibitemOpen
  \bibfield  {author} {\bibinfo {author} {\bibfnamefont {W.}~\bibnamefont
  {Stein}},\ }\href {\doibase 10.1007/b13279} {\emph {\bibinfo {title}
  {Elementary Number Theory: Primes, Congruences, and Secrets}}}\ (\bibinfo
  {publisher} {Springer},\ \bibinfo {address} {New York},\ \bibinfo {year}
  {2008})\BibitemShut {NoStop}%
\end{thebibliography}
\end{document}